\documentclass[11pt, a4 paper]{article}

\usepackage{amin}

\usepackage[margin=1in]{geometry}
\usepackage{fancyhdr}

\title{An Efficient Algorithm Computing Composition Factors of $T(V)^{\otimes n}$}
\author{Amin Saied}
\date{}

\begin{document}
\maketitle

\abstract{
\noindent We present an algorithm that computes the composition factors of the $n$-th tensor power of the free associative algebra on a vector space. The composition factors admit a description in terms of certain coefficients $c_{\lambda\mu}$ determining their irreducible structure. By reinterpreting these coefficients as counting the number of ways to solve certain `decomposition-puzzles' we are able to design an efficient algorithm extending the range of computation by a factor of over 750. Furthermore, by visualising the data appropriately, we gain insights into the nature of the coefficients leading to the development of a new representation theoretic framework called $\PD$-modules.
}

\section{Introduction}
Certain coefficients $c_{\lambda\mu} \in \mathbb{N}$, indexed by pairs of partitions $\lambda, \mu$, naturally arise in the study of the Johnson homomorphism of the mapping class group. They can be thought of as describing the decomposition of Schur functors on the free Lie algebra $\mathcal{L}(V)$ into Schur functors on $V$ itself (see Section \ref{SecMathBackground}). In this paper we present an algorithm computing the coefficients $c_{\lambda\mu}$. Our approach is to reinterpret the coefficients as counting solutions to a certain combinatorial problem we call \emph{decomposition puzzles}. In particular, we prove the following theorem.

\paragraph{Theorem \ref{TheoremPuzzles}.} \emph{The coefficients $c_{\lambda\mu}$ counts the number of (weighted) solutions to $(\mu, \lambda)$ decomposition puzzles.}\\

%This perspective offers a number of advantages. First, the combinatorial description of the coefficients lends itself well to an algorithmic approach. Second, it provides a discretisation of the problem into several steps.

\noindent This combinatorial description provides a discretisation of the problem into several steps outlined in Fig. \ref{FigSolutionPath}. By analysing the computational complexity of each step we are able to make key optimisations to the algorithm. In so doing we are able to compute $257,049$ coefficients, extending the known range of coefficients by a factor of over 750.\\

\noindent At a high level, a solution to a $(\mu, \lambda)$ decomposition puzzle can be represented as a path from $\mu$ to $\lambda$.
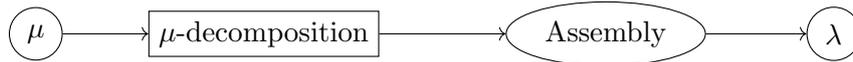
\begin{figure}[H]
\begin{center}
\vspace{-1cm}
\begin{tikzpicture}
[level distance=1.5cm,
level 1/.style={sibling distance=1.5cm},
level 2/.style={sibling distance=1cm},
level 3/.style={sibling distance=0.5cm, level distance = 1cm}]

\tikzstyle{every node}=[circle,draw]

\node [draw=none] (Root) {}
child { node (A) {$\mu$} edge from parent[draw=none] }
child { node [draw=none] {} edge from parent[draw=none] }
child { node [shape=rectangle] (B) {$\mu$-decomposition} edge from parent[draw=none] }
child { node [draw=none] {} edge from parent[draw=none] }
child { node [draw=none] {} edge from parent[draw=none] }
child { node [shape=ellipse] (C) {Assembly} edge from parent[draw=none] }
child { node [draw=none] {} edge from parent[draw=none] }
child { node (D) {$\lambda$} edge from parent[draw=none] }
;

\draw[->] (A) -- (B);
\draw[->] (B) -- (C);
\draw[->] (C) -- (D);

\end{tikzpicture}
\end{center}
\caption{Path representing the steps involved in solving a $(\mu, \lambda)$ decomposition puzzle.}
\label{FigSolutionPath}
\end{figure}

\noindent We collect all such paths into a tree (see Fig. \ref{FigDecompositionTree}), whence Theorem \ref{TheoremPuzzles} reinterprets $c_{\lambda\mu}$ as counting the number of its leaves labelled by $\lambda$. The major hurdle in computing $c_{\lambda\mu}$ is a combinatorial explosion arising in the number of possible assemblies of a given $\mu$-decomposition as the size of $\mu$ grows (Eq. \ref{EqTooManyComputes}). Our key optimisation is the so called \emph{shape analysis} of a $\mu$-decomposition (Section \ref{SubSecShapeAnal}) which allows us to more efficiently search the leaves of the tree by fixing the degree of the target partition $\lambda$ in question.

\newpage

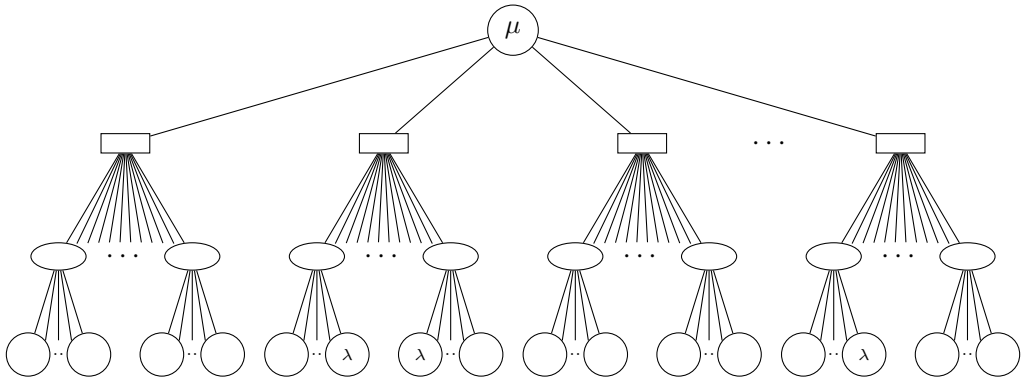
\begin{figure}
\begin{center}
\begin{tikzpicture}[level distance=1.5cm,
level 1/.style={sibling distance=3.4cm},
level 2/.style={sibling distance=0.16cm},
level 3/.style={sibling distance=0.2cm, level distance = 1.3cm}]

\tikzstyle{every node}=[circle,draw]

\node (Root) {\small$\mu$}
child { 
    node [shape=rectangle] {~~~~}
    	child {
		node [shape=ellipse] (ass11) {~~~}
			child{
				node (L11) {~~~}
			}
			child { node[draw=none] {}}
			child { node[draw=none] {}}
			child { node[draw=none] {}}
			child{
				node (L12) {~~~}
			}
	}
	child { node[draw=none] {}}
	child { node[draw=none] {}}
	child { node[draw=none] {}}
	child { node[draw=none] {}}
	child { node[draw=none] {}}
	child { node[draw=none] {}}
	child { node[draw=none] {}}
	child { node[draw=none] {}}
	child { node[draw=none] {}}
	child { node[draw=none] {}}
	child {
		node [shape=ellipse] (ass12) {~~~}
			child{
				node (L21) {~~~}
			}
			child { node[draw=none] {}}
			child { node[draw=none] {}}
			child { node[draw=none] {}}
			child{
				node (L22) {~~~}
			}
	}
}
child {
    node [shape=rectangle]  {~~~~} 
    	child {
		node [shape=ellipse] (ass21) {~~~}
			child{
				node (L31) {~~~}
			}
			child { node[draw=none] {}}
			child { node[draw=none] {}}
			child { node[draw=none] {}}
			child{
				node (L32) {\tiny$\lambda$}
			}
	}
	child { node[draw=none] {}}
	child { node[draw=none] {}}
	child { node[draw=none] {}}
	child { node[draw=none] {}}
	child { node[draw=none] {}}
	child { node[draw=none] {}}
	child { node[draw=none] {}}
	child { node[draw=none] {}}
	child { node[draw=none] {}}
	child { node[draw=none] {}}
	child {
		node [shape=ellipse] (ass22) {~~~}
			child{
				node (L41) {\tiny$\lambda$}
			}
			child { node[draw=none] {}}
			child { node[draw=none] {}}
			child { node[draw=none] {}}
			child{
				node (L42) {~~~}
			}
	}
}
child {
    node [shape=rectangle] (dec1) {~~~~}
    	child {
		node [shape=ellipse] (ass31) {~~~}
			child{
				node (L51) {~~~}
			}
			child { node[draw=none] {}}
			child { node[draw=none] {}}
			child { node[draw=none] {}}
			child{
				node (L52) {~~~}
			}
	}
	child { node[draw=none] {}}
	child { node[draw=none] {}}
	child { node[draw=none] {}}
	child { node[draw=none] {}}
	child { node[draw=none] {}}
	child { node[draw=none] {}}
	child { node[draw=none] {}}
	child { node[draw=none] {}}
	child { node[draw=none] {}}
	child { node[draw=none] {}}
	child {
		node [shape=ellipse] (ass32) {~~~}
			child{
				node (L61) {~~~}
			}
			child { node[draw=none] {}}
			child { node[draw=none] {}}
			child { node[draw=none] {}}
			child{
				node (L62) {~~~}
			}
	}
}
child {
    node [shape=rectangle] (dec2) {~~~~}
    	child {
		node [shape=ellipse] (ass41) {~~~}
			child{
				node (L71) {~~~}
			}
			child { node[draw=none] {}}
			child { node[draw=none] {}}
			child { node[draw=none] {}}
			child{
				node (L72) {\tiny$\lambda$}
			}
	}
	child { node[draw=none] {}}
	child { node[draw=none] {}}
	child { node[draw=none] {}}
	child { node[draw=none] {}}
	child { node[draw=none] {}}
	child { node[draw=none] {}}
	child { node[draw=none] {}}
	child { node[draw=none] {}}
	child { node[draw=none] {}}
	child { node[draw=none] {}}
	child {
		node [shape=ellipse] (ass42) {~~~}
			child{
				node (L81) {~~~}
			}
			child { node[draw=none] {}}
			child { node[draw=none] {}}
			child { node[draw=none] {}}
			child{
				node (L82) {~~~}
			}
	}
};

\path (dec1.west) -- (dec2.east) node[draw=none] [midway] {$\cdots$};
\path (ass11.west) -- (ass12.east) node[draw=none] [midway] {$\cdots$};
\path (ass21.west) -- (ass22.east) node[draw=none] [midway] {$\cdots$};
\path (ass31.west) -- (ass32.east) node[draw=none] [midway] {$\cdots$};
\path (ass41.west) -- (ass42.east) node[draw=none] [midway] {$\cdots$};

\path (L11.west) -- (L12.east) node[draw=none] [midway] {\tiny$\cdot\cdot$};
\path (L21.west) -- (L22.east) node[draw=none] [midway] {\tiny$\cdot\cdot$};
\path (L31.west) -- (L32.east) node[draw=none] [midway] {\tiny$\cdot\cdot$};
\path (L41.west) -- (L42.east) node[draw=none] [midway] {\tiny$\cdot\cdot$};
\path (L51.west) -- (L52.east) node[draw=none] [midway] {\tiny$\cdot\cdot$};
\path (L61.west) -- (L62.east) node[draw=none] [midway] {\tiny$\cdot\cdot$};
\path (L71.west) -- (L72.east) node[draw=none] [midway] {\tiny$\cdot\cdot$};
\path (L81.west) -- (L82.east) node[draw=none] [midway] {\tiny$\cdot\cdot$};

\end{tikzpicture}
\end{center}
\caption{Tree representation of all decomposition puzzles. The root here is labelled by the partition $\mu$ and Theorem \ref{TheoremPuzzles} reinterprets $c_{\lambda\mu}$ as counting the number of leaves labelled by $\lambda$. Shape analysis prunes unwanted assemblies (nodes at depth 2), avoiding much unwanted computation.}
\label{FigDecompositionTree}
\end{figure}

\noindent With the algorithm in hand, we turn to the analysis of the data. Visualising the data appropriately we notice clustering patterns among the coefficients (as in Fig. \ref{FigVisOverview} below). We study these patterns in Section \ref{SecDataAnalysis} and present several conjectures about the combinatorial structure among the coefficients. 

\begin{figure}[H]
\begin{center}
\input{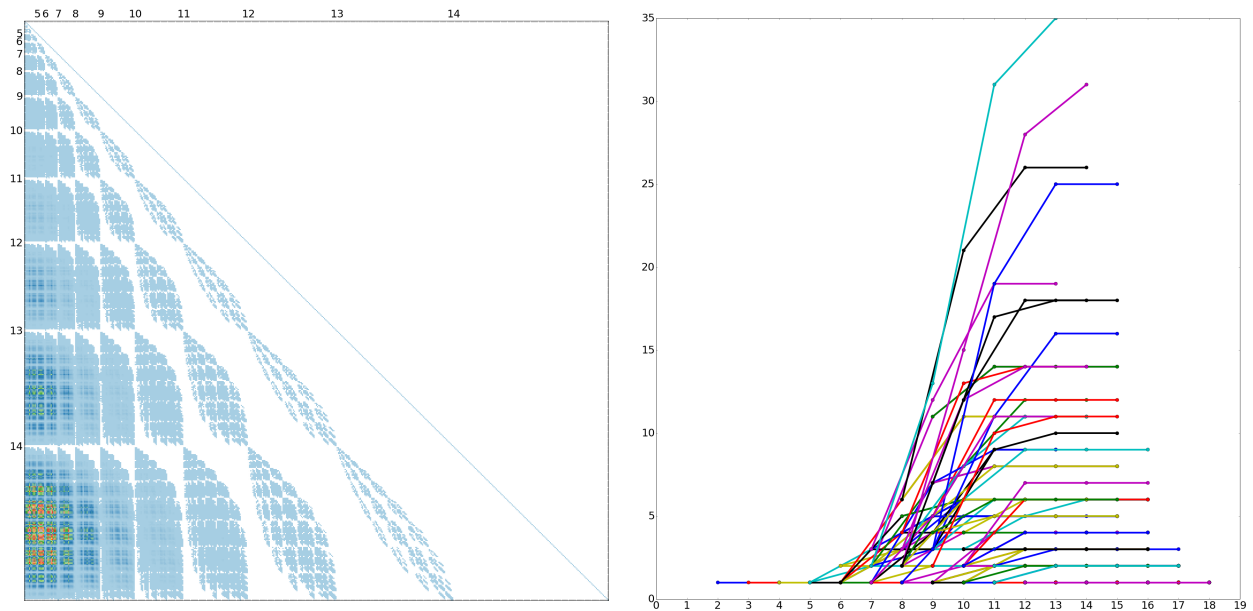}
\caption{Two visualisations of the coefficient data computed by our algorithm. The patterns emerging from these visualisations are suggestive of a stability phenomenon akin to $\FI$-modules. The plot on the left is a representation of all 257,049 coefficients computed by our algorithm. The plot on the right shows the evolution of those coefficients in a certain \emph{stable direction}. See Section \ref{SecDataAnalysis} for more detail.}
\label{FigVisOverview}
\end{center}
\end{figure}

\noindent One striking observation is a stability pattern akin to the representation stability of Church-Farb-Ellenberg \cite{CF, CEF}. In particular, their theory of $\FI$-modules has strong parallels with the stability patterns that emerge from our coefficient data. It is these parallels that lead to the conjecture of a new representation theoretic framework in the mould of $\FI$-modules. We introduce this representation theoretic framework, called the theory of $\PD$-modules, in a forthcoming paper \cite{Amin2}.

\newpage

\paragraph{Outline.} In Section \ref{SecPrelims} we provide some background material needed throughout the paper. In Section \ref{SecDecompPuzz} we introduce decomposition puzzles and prove the above theorem. The key insight behind our algorithm comes in Section \ref{SubSecShapeAnal} - which we call \emph{shape analysis}. We give the algorithm in Section \ref{SecAlgorithm}. In Section \ref{SecDataAnalysis} we analyse the data, giving several visualisations. These will lead to a number of conjectures relating to the combinatorics of the coefficients. Finally in Section \ref{SecRunningTime} we analyse the running time of our algorithm against a baseline algorithm.\\

\noindent The source code for our algorithm is publicly available on GitHub:
\begin{center}
\url{https://github.com/aminsaied/composition_factors}
\end{center}

\paragraph{Acknowledgements.} The author thanks his thesis advisor Martin Kassabov for introducing the problem and for his many invaluable insights.

\section{Preliminaries} \label{SecPrelims}

Throughout this paper we will make frequent use of Young diagrams to depict partitions. Concretely, a partition $\lambda \vdash n$ is a weakly decreasing sequence of positive integers $(\lambda_1, \ldots, \lambda_k)$ where $\sum_i \lambda_i = n$. To such a partition we associate a Young diagram, which is a collection of left-justified boxes, with $\lambda_i$ boxes in row $i$. For example, $\lambda = (5,3,2,1) \vdash 11$ has the following Young diagram.

\[
\Yvcentermath1
\yng(5,3,2,1)
\]

\noindent We also make use of the well-known correspondence between partitions of $n$ and irreducible representation of the symmetric group $\SYM_n$. We denote the irreducible $\SYM_n$-module associated to the partition $\lambda$ by $P_\lambda$, or sometimes simply by the Young diagram for $\lambda$.\\

\noindent Given a vector space $V$ and a partition $\lambda$, the vector space,
\[
\SF{\lambda}(V):= V^{\otimes n} \otimes_{\SYM_n} P_\lambda,
\]
is an irreducible $\GL(V)$-module. We call the functor $\SF{\lambda}$ the Schur functor (associated to $\lambda$). We will also have occasion to denote Schur functors $\SF{\lambda}(V)$ simply by their underlying Young diagrams when the distinction is clear or unimportant.\\

\noindent Two key ingredients baked into our algorithm are the \emph{tensor product} of partitions and the \emph{plethysm} of partitions. We describe those constructions now.

\subsection{Tensor product, $\otimes$}

We define the tensor product of two partitions $\lambda, \mu$ in terms of the Littlewood-Richardson coefficients. That is to say, the tensor product $\lambda \otimes \mu$ is a set of partitions $\nu$ of size $|\lambda| + |\mu|$ with multiplicities determined by the Littlewood-Richardson coefficient $L^\nu_{\lambda\mu}$. Concretely,
\begin{equation}
\lambda \otimes \mu = \bigoplus_{\nu \vdash |\lambda|+|\mu|} \nu^{\oplus L^{\nu}_{\lambda\mu}}
\end{equation}

\noindent The Littlewood-Richardson coefficients are \emph{combinatorial} in the sense that there is a combinatorial rule, known as the Littlewood-Richardson rule, for computing the coefficients. It turns out that the Littlewood-Richardson coefficients can be interpreted as counting the number of solutions of certain puzzles, so-called `honeycombs', introduced by Knutson-Tao-Woodward in \cite{KnutsonTaoWoodward}.\\

\begin{remark}
An important fact about the tensor product of partitions is that they sums sizes, so that if $\nu$ appears in $\lambda \otimes \mu$ then $|\nu| = |\lambda| +|\mu|$.
\end{remark}

\begin{example} The tensor product of $[2] \vdash 2$ with $[2,1] \vdash 3$ is a sum of partitions of size $2+3=5$.
\[
\Yvcentermath1
\yng(2) ~\otimes~ \yng(2,1) ~=~ \yng(2,2,1) ~\oplus~ \yng(3,1,1) ~\oplus~ \yng(3,2) ~\oplus~ \yng(4,1)
\]
\end{example}

\subsection{Plethysm, $\odot$}

The story for the plethysm is similar, but more complicated. The original plethysm problem is to understand the coefficients $M^\nu_{\lambda\mu}$ describing the composition of Schur functors,
\[
\SF{\mu} ( \SF{\lambda} (V) ) = \bigoplus_{\nu \vdash |\lambda|\cdot|\mu|}  \SF{\nu}(V)^{\oplus M^\nu_{\lambda \mu}}
\] 
We can then define the plethysm of two partitions $\lambda, \mu$ in terms of the coefficients $M^\nu_{\lambda\mu}$ via the formula,
\begin{equation}
\lambda \odot \mu = \bigoplus_{\nu \vdash |\lambda|\cdot|\mu|} \nu^{\oplus M^{\nu}_{\lambda\mu}}.
\end{equation}

\noindent The coefficients $M^\nu_{\lambda \mu}$ are much less well understood. We refer the reader to Stanley (\cite{Stanley1}) and Fulton-Harris (\cite{FH}) for an introduction.

\begin{remark}~
\begin{enumerate}
\item
An important fact about the plethysm\footnote{The word plethysm is from the Greek work meaning `multiplication'.} of partitions  is that it multiplies sizes, so that if $\nu$ appears in the $\lambda \odot \mu$ then $|\nu| = |\lambda| \cdot |\mu|$.

\item
The plethysm of symmetric functions is implemented in SAGE, where we implement our algorithm.
\end{enumerate}
\end{remark}

\begin{example} The plethysm of $[2] \vdash 2$ with $[2,1] \vdash 3$ is a sum of partitions of size $2\cdot3=6$.
\[
\Yvcentermath1
\yng(2) ~\odot~ \yng(2,1) ~=~ \yng(2,2,2) ~\oplus~ \yng(3,1,1,1) ~\oplus~ \yng(3,2,1) ~\oplus~ \yng(4,2)
\]
\end{example}

\subsection{Coefficients arising in the study of $T(V)^{\otimes n}$} \label{SecMathBackground}
In a forthcoming paper \cite{CKS} we study the connection between the Johnson homomorphism of the mapping class group and a certain $\Aut(F_n)$-module $T(V)^{\otimes n}$. In this section we briefly outline where the coefficients $c_{\lambda\mu}$ arise in that study.\\

%We show how the coefficients $c_{\lambda\mu}$ arise in the study of the $\Aut(F_n)$-module $T(V)^{\otimes n}$. This section should serve as motivation for computing the coefficients $c_{\lambda\mu}$, the task to which the rest of our efforts will be devoted.\\

\noindent Recall that the tensor algebra $T(V)$ can be viewed as the universal enveloping algebra of the free Lie algebra $\mathcal{L}(V)$ (see Eq (\ref{EqFreeLieAlgebra})) and as such has an increasing filtration $T(V)^{(i)}$. This induces a filtration on $T(V)^{\otimes n}$, and the PWB theorem gives that the associated graded
\[
\gr T(V)^{\otimes n} \cong \Sym^\ast(\mathcal{L}(V) \otimes \F^n ).
\]
The RHS admits a decomposition by Schur functors (see, for example, Fulton-Harris \cite{FH}), and we have,
\[
\Sym^\ast(\mathcal{L}(V) \otimes \F^n ) \cong \bigoplus_{\mu} \SF{\mu}\left( \mathcal{L}(V) \right) \otimes \SF{\mu}(\F^n).
\]
We introduce the coefficients $c_{\lambda\mu}$ by expressing the Schur functors on $\mathcal{L}(V)$ in terms of Schur functors on $V$, giving,

\begin{equation} \label{EqCoeffsFromFreeLie}
\SF{\mu}\left( \mathcal{L}(V) \right) \cong \bigoplus_{\lambda} c_{\lambda\mu} \SF{\lambda}(V),
\end{equation}
an infinite sum over all partitions.\\

\noindent We take (\ref{EqCoeffsFromFreeLie}) as a definition for the coefficient $c_{\lambda\mu}$. In other words, we define $c_{\lambda\mu}$ as the number of times $\SF{\lambda}(V)$ appears in the decomposition of $\SF{\mu}(\mathcal{L}(V))$.

\begin{remark}
In \cite{CKS} we expand on the relationship between the coefficients $c_{\lambda\mu}$ and the cohomology $H^\ast(\Aut(F_n);T(V)^{\otimes n})$ appearing in the study of the Johnson homomorphism of the mapping class group. In particular, we are able to make some cokernel computations in rank 2 and rank 3.
\end{remark}

\noindent The description of $c_{\lambda\mu}$ in (\ref{EqCoeffsFromFreeLie}) ties the coefficients to the study of the Johnson homomorphism of the mapping class group. In what follows, however, we are only interested in computing the coefficients themselves. We therefore devote the next section to recasting the definition of $c_{\lambda\mu}$ in combinatorial terms well suited to an algorithmic approach. The connection with the above description of $c_{\lambda\mu}$ in (\ref{EqCoeffsFromFreeLie}) is given in Theorem \ref{TheoremPuzzles}.

%\noindent The remainder of this paper will be devoted to computing the coefficients $c_{\lambda\mu}$. Our strategy will be to reinterpret this expression in a combinatorial way that lends itself well to computation. In particular, we will turn (\ref{EqCoeffsFromFreeLie}) into a combinatorial \emph{puzzle} that we describe now.

\section{Decomposition puzzles} \label{SecDecompPuzz}
In the introduction we represented a decomposition puzzle as a path in a certain tree. We start by expanding that path into a schematic overview of decomposition puzzles. We will go on to describe the component moves in the remainder of the section.

\begin{figure}[H]
\begin{center}
\begin{tikzpicture}[scale=0.8]
\Yvcentermath1

% mu partition
\node at (0,5) {$\mu = \yng(3,2)$};

% lambda partition
\node at (15,5) {$\lambda = \yng(4,3,2,1)$};

% decomposition
\node at (5,0) {$\yng(2)$};
\node at (5,-2) {$\yng(1,1)$};
\node at (5,-4) {$\yng(1)$};

% free lie algebra
\node at (11.2,2) [rectangle,draw,fill=black!20,minimum height=3em, minimum width=6em] {};
\node at (11.2,2) [align=center] {Lie pieces \\ \small Section \ref{SubSecLiePieces}};
\node at (10,0) {$\yng(1)$};
\node at (10,-2) {$\yng(1,1)$};
\node at (10,-4) {$\yng(2,1)$};
\node at (10,-6) {$\yng(2,1,1)$};
\node at (10,-8) {$\vdots$};

%% arrows
% mu-decomposition arrow
\draw[->, thick]  (0,4)  -- (0,-2) node [below left] {} -- (4,-2);
\node at (0,-2) [rectangle,draw,fill=black!20,minimum height=3.3em, minimum width=8.2em] {};
\node at (0,-2) [align=center] {$\mu$-decomposition \\ \small Section \ref{SubSecMuDecomposition}};
% lambda-arrow
\draw[->, thick]  (11,-2)  -- (15,-2) -- (15,3.5);
% c_lambda,mu arrow
\draw[->, dashed] (2,5) -- (7.5,5) node [above] {$c_{\lambda\mu}$-contribution} -- (13,5);

%assembly arrows
\draw[thick] (5.8, 0) -- (6.2, 0) -- (8.8,-4) -- (9.2, -4);
\draw[thick] (5.8, -2) -- (6.2, -2) -- (8.8,0) -- (9.2, 0);
\draw[thick] (5.8, -4) -- (6.2, -4) -- (8.8,-2) -- (9.2, -2);

% assembly box
%\draw[thick, fill=black!20, fill opacity=0.7]  (6,2) -- (9,2) -- (9,-6) -- (6,-6) -- cycle;
\draw[thick]  (6,2) -- (9,2) -- (9,-6) -- (6,-6) -- cycle;
\node at (7.5,3) [rectangle,draw,fill=black!20,minimum height=3em, minimum width=6.0em] {};
\node at (7.5,3) [align=center] {Assembly \\ \small Section \ref{SubSecAssembly}};

\end{tikzpicture}

%\tikzstyle{stuff_fill}=[rectangle,draw,fill=white,minimum height=2em, minimum width=4em,label={center:Amin}]
%\draw (-0.5,2) circle (0.45) node {\Large $\Sigma$};
\caption{A schematic overview of a solution to a decomposition puzzle. This also serves as an example, where in this case we have a partition $\mu = (3,2)$ decomposing into three partitions $(2), (1,1)$ and $(1)$. These are then assembled with Lie pieces when finally we arrive at the partition $\lambda=(4,3,2,1)$}
\label{TikzOverview}
\end{center}
\end{figure}
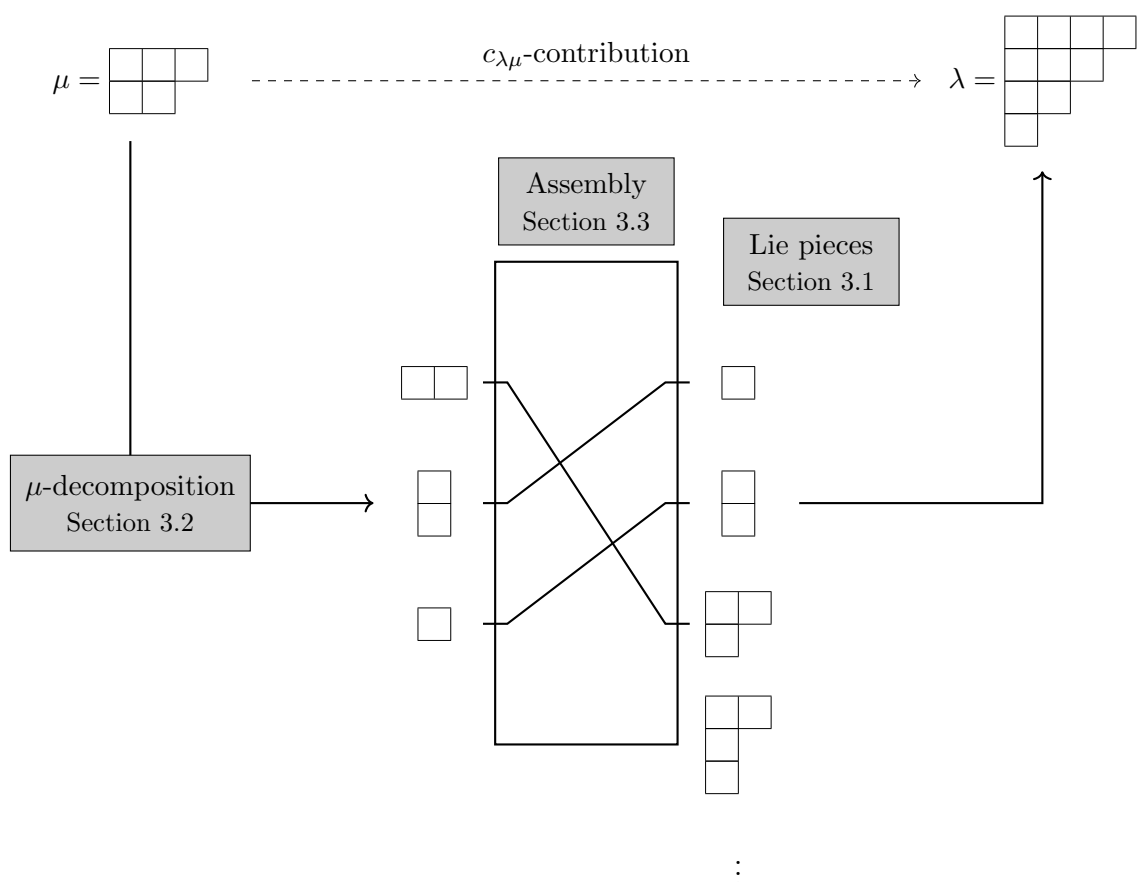

\newpage
\subsection{Lie pieces} \label{SubSecLiePieces}
Central to this point of view is the decomposition of the free Lie algebra into its irreducible $\GL(V)$-modules. Given a vector space $V$, the free Lie algebra on $V$ is a graded vector space,

\begin{equation} \label{EqFreeLieAlgebra}
\mathcal{L}(V) = \bigoplus_{i \geq 0} \mathcal{L}_i(V),
\end{equation}
whose graded pieces $\mathcal{L}_i(V)$ are $GL(V)$-modules. The decomposition into irreducible $\GL(V)$-modules of the first few terms are listed below.
\[
\Yvcentermath1
\mathcal{L}_1(V) = \yng(1) \qquad \mathcal{L}_2(V) = \yng(1,1) \qquad  \mathcal{L}_3(V) = \yng(2,1)
\]
\[
\Yvcentermath1
\mathcal{L}_4(V) = \yng(3,1) ~\oplus~ \yng(2,1,1)
\]
\[
\Yvcentermath1
\mathcal{L}_5(V) = \yng(4,1) ~\oplus~ \yng(3,2) ~\oplus~ \yng(3,1,1) ~\oplus~ \yng(2,2,1) ~\oplus~ \yng(2,1,1,1)
\]
In general we can describe the $i$-th term $\mathcal{L}_i(V)$ as,
\[
\mathcal{L}_i(V) \cong \Lie_i \otimes V^{\otimes i},
\]
where $\Lie_i$ is an $\SYM_i$-module Schur-Weyl dual to $\mathcal{L}_i(V)$ known as the Whitehouse module\footnote{It is also the arity $i-1$ part of the Lie operad.}. There is a combinatorial rule describing its irreducible decomposition by counting certain standard Young tableaux.

\begin{definition}
A standard \textbf{Young tableaux} of shape $\lambda \vdash n$ is a Young diagram of shape $\lambda$ filled in (bijectively) with the numbers $\{1, \ldots, n\}$ so that the numbers are increasing along the rows and columns.
\end{definition}

\begin{definition}
Given a tableaux $T$ of shape $\lambda$, define $\mbox{maj}(T)$ as the sum of $i$ such that $i+1$ lies below $i$ in $T$.
\end{definition}

\begin{example}
Let $\lambda = (2,1,1) \vdash 4$, then,
\[
T=\ytableausetup{centertableaux}
\begin{ytableau}
 1&2\\
3\\
4
\end{ytableau}
\]
is a standard tableaux of shape $\lambda$. We have that $\mbox{maj}(T) = 2 + 3 = 5$.
\end{example}

\begin{theorem}[Stanley]
Let $\lambda \vdash d$. Then the multiplicity of $P_\lambda$ in $\Lie_d$ is given by the number of Young tableaux $T$ of shape $\lambda$ satisfying $\mbox{maj}(T) \equiv 1 \mod d$.
\end{theorem}

\noindent This theorem governs the partitions $\lambda$ appearing in the irreducible decomposition of the Whitehouse modules $\Lie_d$ for all $d>0$. Moreover, it gives the multiplicity with which each partition appears. We collect all such partitions, counted with multiplicity into an (infinite) collection $\mathbb{L}$ of Lie pieces. Consequently, we can use $\mathbb{L}$ to describe the Whitehouse modules $\Lie_d$ and the free Lie algebra $\mathcal{L}_d(V)$.

\begin{equation} \label{EqLiePiecesDecomp}
\Lie_d = \bigoplus_{\substack{ \lambda \vdash d \\ \lambda \in \mathbb{L}} } P_\lambda \qquad \qquad \mathcal{L}_d(V) = \bigoplus_{\substack{ \lambda \vdash d \\ \lambda \in \mathbb{L}} } \SF{\lambda}(V)
\end{equation}

\begin{definition}
A \textbf{Lie piece} is a Young diagram appearing in $\mathbb{L}$.
\end{definition}

\newpage
\begin{remark}~
\begin{enumerate}
\item 
It is important to note that there are duplicates in the collection of Lie pieces. For example, the term $\SF{[3,2,1]}(V)$ appears with multiplicity 3 in $\mathcal{L}_6(V)$, so there are three copies of
\[
\Yvcentermath1
\yng(3,2,1)
\]
in the collection of Lie pieces.

\item
It will be convenient in what follows to fix, once and for all, an order on $\mathbb{L}$. We order the pieces first in increasing size order. If Lie pieces are of the same size then we order the partitions lexicographically (lex order), putting those partitions with largest lex order first. We list the first few terms in $\mathbb{L}$.

\[
\Yvcentermath1
\begin{array}{c|ccccccccc}
\mbox{Index} & 1 & 2 & 3 & 4 & 5 & 6 & 7 & 8\\
& & & & & & & &\\
\mbox{Lie pieces} & \yng(1) & \yng(1,1) & \yng(2,1) & \yng(3,1) & \yng(2,1,1) & \yng(4,1) & \yng(3,2) & \yng(3,1,1)
\end{array}
\]

%\item
%We implement a class \texttt{Lie} that represents the free Lie algebra. In particular it computes the Lie pieces and stores them in the order described above. Our source code is publicly available on GitHub\footnote{\url{https://github.com/aminsaied/composition_factors}}.

\end{enumerate}
\end{remark}

\noindent The free Lie algebra is an infinite-dimensional vector space, a fact which does not lend itself well to the kinds of finite computation we are interested in here. In practice we therefore work with a truncated, finite-dimensional piece of the free Lie algebra.

\begin{definition}(Truncation.)
The truncation (of degree $d$) of the free Lie algebra is,
\[
\mathcal{L}_{\leq d}(V) := \bigoplus_{i \leq d} \mathcal{L}_i(V).
\]
The truncation of Lie pieces, denoted $\mathbb{L}_{\leq d}$, is the subcollection of $\mathbb{L}$ consisting of Young diagrams with size at most $d$.
\end{definition}

\begin{remark}
The truncation $\mathcal{L}_{\leq d}(V)$ is also known as the \emph{free} $d$-step nilpotent Lie algebra on $V$.
\end{remark}

\begin{remark} \label{RemarkLieGrowsFast}
We point out that the number of Lie pieces in $\mathbb{L}_{\leq d}$ grows rapidly as a function of $d$. Here are the sizes of the first ten truncations.
\[
\begin{array}{c|cccccccccc}
d & 1 & 2 & 3 & 4 & 5 & 6 & 7 & 8 & 9 & 10 \\\hline |\mathbb{L}_{\leq d}| & 1 & 2 & 3 & 5 & 10 & 22 & 55 & 149 & 439 & 1388
\end{array}
\]
It is the rapid growth indicated here that causes the dramatic slowdown in computing $c_{\lambda\mu}$ for partitions $\lambda, \mu$ of large degree (see (\ref{EqTooManyComputes}) for example).
\end{remark}

\subsection{$\mu$-decompositions} \label{SubSecMuDecomposition}

\begin{definition} \label{DefMuDecomp}
Let $\mu$ be a partition. A \textbf{$\mu$-decomposition} is a collection of (not necessarily distinct) partitions $(\mu_1, \ldots, \mu_k)$ such that,
\[
|\mu| = |\mu_1| + \cdots + |\mu_k|.
\]
We consider two $\mu$-decompositions $(\mu_1, \ldots, \mu_k), (\mu'_1, \ldots, \mu'_l)$ equivalent if $k=l$ and there exists some permutation of the indices $\sigma \in \SYM_k$ such that the ordered collections agree:
\[
(\mu'_1, \ldots, \mu'_k) = (\mu_{\sigma(1)}, \ldots, \mu_{\sigma(k)}).
\]
We tacitly impose this equivalence relation, and choose representatives of equivalence classes as those $\mu$-decompositions $(\mu_1, \ldots, \mu_k)$ where $|\mu_i| \geq |\mu_{i+1}|$, and if $|\mu_i| = |\mu_{i+1}|$ we order them lexicographically.
\end{definition}

\subsubsection{Iterated Littlewood-Richardson coefficients}
\begin{definition}
The iterated Littlewood-Richardson coefficient $L^\mu_{\mu_1\cdots\mu_k}$ of a partition $\mu$ and a $k$-tuple of partitions $(\mu_1, \ldots, \mu_k)$ is defined, for $k>2$, in terms of usual Littlewood-Richardson coefficients as,
\begin{equation}
L^\mu_{\mu_1\cdots\mu_k} := \sum_{\nu_1, \ldots, \nu_{k-2}}L^\mu_{\mu_1\nu_1}L^{\nu_1}_{\mu_2\nu_2}\cdots L^{\nu_{k-2}}_{\mu_{k-1}\mu_{k}},
\end{equation}
where $\nu_i$ are partitions with sizes given below:
\begin{enumerate}
\item $|\nu_1| = |\mu|-|\mu_1|$
\item $|\nu_i|=|\nu_{i-1}|-|\mu_i|$ for $2 \leq i\leq k-2$
\end{enumerate}
For convenience we extend the definition to collections of  size $k = 1, 2$ by declaring that the coefficient $L^\mu_{\mu_1\mu_2}$ is the usual Littlewood-Richardson coefficient, and that the coefficient $L^\mu_{\mu_1}$ is the indicator function on the partition $\mu$.
\end{definition}

\begin{definition}
We say a $\mu$-decomposition $(\mu_1, \ldots, \mu_k)$ is \textbf{good} if $L^\mu_{\mu_1\cdots\mu_k} > 0$.
\end{definition}

\noindent There is a recursive algorithm computing these iterated Littlewood-Richardson coefficients, and thus determining if a given $\mu$-partition is good.\\

%%%%%%%% ALGORITHM %%%%%%%%%%%
\begin{minipage}{.9\linewidth}

\begin{algorithm}[H]

\label{AlgIterLR}

\SetKwFunction{IterLR}{\textsc{iter$\_$lr}}

\SetAlgoLined

%\KwInput{A partition $\mu$ and a set of $k$ partitions $\{\mu_0, \ldots, \mu_{k-1}\}$}\;
%\KwResult{$L^\mu_{\mu_0\cdots\mu_{k-1}}$ }\;

\BlankLine

\KwInput{A partition $\mu$ and an array $D$ of $k$ partitions $[\mu_1, \ldots, \mu_k]$.}

\BlankLine

\KwResult{ Return the iterated Littlewood-Richardson coefficient $L^\mu_{\mu_1\cdots\mu_k}$.}

\BlankLine

\Indm\nonl\IterLR{$\mu,D$}\\
\Indp

\uIf{length $D = 1$}{

$p \gets D[0]$

\Return Indicator $\mathbb{I}_\mu(p)$
}

\BlankLine

\uElseIf{length $D = 2$}{

$p, q \gets D[0], D[1]$

\Return $L^\mu_{p, q}$
}

\BlankLine

\Else{
$p \gets D[0]$

$ m \leftarrow |\mu| - |p|$ \;

$ c \leftarrow 0$ \;

\For{$\nu \in \{ \nu \vdash m: \nu \subseteq \mu \}$}{

$l = L^{\mu}_{p\nu}$\;

\uIf{$l > 0$}{
x = \IterLR{$\nu, D[1: ]$}

$c \leftarrow c + l\ast x$
}

}
\Return $c$
}

\caption{Iterated Littlewood-Richardson Coefficient}

\end{algorithm}

\end{minipage}
%%%%%%%% ALGORITHM %%%%%%%%%%%

\subsection{Assembly} \label{SubSecAssembly}

We now describe assembly; the process by which partitions $\lambda$ are constructed from a $\mu$-decomposition and a tuple of Lie pieces.\\

\begin{definition}
A \textbf{pairing} of a $\mu$-decomposition $D_\mu = (\mu_1, \ldots, \mu_k)$ is a collection of $k$ distinct\footnote{Distinct indices of Lie pieces, as opposed to distinct partitions. The distinction is important as there are multiplicities $>1$ appearing in the decomposition of the free Lie algebra.} Lie pieces $L = (l_{i_1}, \ldots, l_{i_k})$ together with a bijection $\phi$ on the indices of $D_\mu$ and of $L$.
\end{definition}

\noindent For clarity, we consider \emph{straightening} the pairing by relabelling the Lie pieces according to the bijection $\phi$ so that $\mu_j$ is paired with $l_{i_j}$. We denote such a pairing by
\[
(\mu_1, \ldots, \mu_k) \smile (l_{i_1}, \ldots, l_{i_k}).
\]
We depict a pairing, together with its straightened counterpart in Fig. \ref{FigAssembly} below.

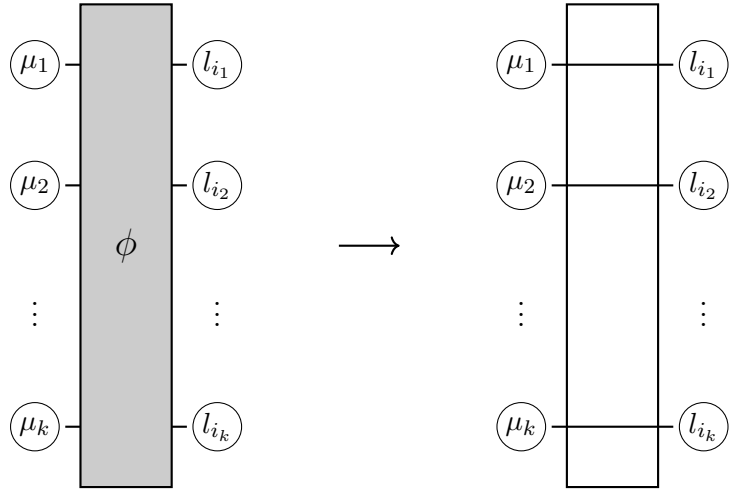
\begin{figure}[H]
\begin{center}
\begin{tikzpicture}[scale=0.8]
\Yvcentermath1

% mus
\draw (0,6) circle (0.4) node {$\mu_1$};
\draw (0,4) circle (0.4) node {$\mu_2$};
\node at (0,2)  {$\vdots$};
\draw (0,0) circle (0.4) node {$\mu_k$};

% lies
\draw (3,6) circle (0.4) node {$l_{i_1}$};
\draw (3,4) circle (0.4) node {$l_{i_2}$};
\node at (3,2)  {$\vdots$};
\draw (3,0) circle (0.4) node {$l_{i_k}$};

% arrows
\draw[thick] (0.5,6) -- (2.5,6);
\draw[thick] (0.5,4) -- (2.5,4);
\draw[thick] (0.5,0) -- (2.5,0);

% box
\draw[thick, fill=black!20]  (0.75,-1) -- (2.25,-1) -- (2.25,7) -- (0.75,7) -- cycle;
\node at (1.5, 3) {\Large $\phi$};

\draw[->, thick] (5, 3) --  (6,3);

\begin{scope}[xshift=8cm]
% mus
\draw (0,6) circle (0.4) node {$\mu_1$};
\draw (0,4) circle (0.4) node {$\mu_2$};
\node at (0,2)  {$\vdots$};
\draw (0,0) circle (0.4) node {$\mu_k$};

% lies
\draw (3,6) circle (0.4) node {$l_{i_1}$};
\draw (3,4) circle (0.4) node {$l_{i_2}$};
\node at (3,2)  {$\vdots$};
\draw (3,0) circle (0.4) node {$l_{i_k}$};

% arrows
\draw[thick] (0.5,6) -- (2.5,6);
\draw[thick] (0.5,4) -- (2.5,4);
\draw[thick] (0.5,0) -- (2.5,0);

% box
\draw[thick]  (0.75,-1) -- (2.25,-1) -- (2.25,7) -- (0.75,7) -- cycle;
%\node at (1.5, 3) {\Large $\phi$};
\end{scope}

\end{tikzpicture}
\caption{On the left we depict a pairing of $(\mu_1, \ldots, \mu_k)$ with a collection of Lie pieces $(l_{i_1}, \ldots, l_{i_k})$. On the right is the straightened version of this pairing, with the indices of the Lie pieces shuffled and relabelled according to the bijection $\phi$.}
\label{FigAssembly}
\end{center}
\end{figure}

\noindent We are now ready to describe the assembly of a (straightened) pairing.

\begin{definition}
An \textbf{assembly}\footnote{Here both senses of the word are employed. On the one hand, we think of \emph{assembling} two collections of partitions, and on the other we think of the \emph{assembled} collection of partitions that arise from the construction.} of a (straightened) pairing,
\[
(\mu_1, \ldots, \mu_k) \smile (l_{i_1}, \ldots, l_{i_k}),
\]
is the collection of partitions arising in,
\begin{equation} \label{EqAssembly}
(\mu_1 \odot l_{i_1}) \otimes (\mu_2 \odot l_{i_2}) \otimes \cdots \otimes (\mu_k \odot l_{i_k}
). \end{equation}
We denote this assembly by,
\[
(\mu_1, \ldots, \mu_k) \oast (l_{i_1}, \ldots, l_{i_k}).
\]
\end{definition}

\begin{remark}
The expression (\ref{EqAssembly}) is where a lot of the work is being done in computing $c_{\lambda\mu}$. Here we iteratively apply plethysms and tensor products of various partitions. When our partitions are relatively small, this can be done quickly, but as our partitions become large enough it becomes infeasible. There is no getting around this fact, and so our goal is to make the minimal number of applications of $\oast$ as possible.
\end{remark}

\noindent The following result forms the basis of our approach to computing $c_{\lambda\mu}$.

\begin{lemma} \label{LemmaAssemblyNoDec}
Fix a partition $\mu \vdash m$, and a $\mu$-decomposition $(\mu_1, \ldots, \mu_k)$. Then any assembly with $(\mu_1, \ldots, \mu_k)$ consists of partitions of size at least $m$. Moreover, if,
\[
(\mu_1, \ldots, \mu_k) \smile (l_{i_1}, \ldots, l_{i_k})
\]
is a pairing, then every partition appearing in its assembly is of size,
\[
|\mu_1|\cdot|l_{i_1}| + \ldots |\mu_k|\cdot|l_{i_k}|.
\]
\end{lemma}

\begin{proof}
The first statement follows immediately from the second. The second is a straightforward consequence of the definition of an assembly as a sequence of plethysms and tensor products.
\end{proof}

\noindent In light of this lemma we make the following definition.

\begin{definition}
We say an assembly $(\mu_1, \ldots, \mu_k) \oast (l_{i_1}, \ldots, l_{i_k})$ has \textbf{target-size},
\[
|\mu_1|\cdot|l_{i_1}| + \ldots |\mu_k|\cdot|l_{i_k}|.
\]
\end{definition}

\begin{example}
We are ready to give an example of a solution to a decomposition puzzle. Let $\mu = [2,1]$. Then an example of a good $\mu$-decomposition is,
\[
\Yvcentermath1
 \mu_1 ~=~ \yng(2) \qquad \mu_2 ~=~ \yng(1)
\]
An example of a (straightened) pairing of this $\mu$-decomposition is,
\[
\Yvcentermath1
 l_2 ~=~ \yng(1,1) \qquad l_1 ~=~ \yng(1)
\]
We compute the corresponding assembly of $(\mu_1, \mu_2) \smile (l_2, l_1)$.
\[
\Yvcentermath1
\mu_1 \odot l_2 ~=~ \yng(2) ~\odot~ \yng(1,1) ~=~ \yng(1,1,1,1) ~\oplus~ \yng(2,2) \qquad \qquad
\mu_2 \odot l_1 ~=~ \yng(1) ~\odot~ \yng(1) ~=~ \yng(1)
\]
\[
\Yvcentermath1
(\mu_1, \mu_2) \oast (l_2, l_1) ~=~  \left(~ \yng(1,1,1,1) ~\oplus~ \yng(2,2)~ \right) ~\otimes~ \yng(1) ~=~ \yng(1,1,1,1,1) ~\oplus~ \yng(2,1,1,1) ~\oplus~ \yng(2,2,1) ~\oplus~ \yng(3,2)
\]
Observe that all partitions appearing in $(\mu_1, \mu_2) \oast (l_2, l_1)$ are of size
\[
5 = |\mu_1| \times |l_2| + |\mu_2| \times |l_1|.
\]
Therefore this assembly has target-size 5. The corresponding paths in the tree in Fig. \ref{FigDecompositionTree} are shown below.

\begin{center}
\begin{tikzpicture}
[level distance=2cm,
level 1/.style={sibling distance=1.5cm},
level 2/.style={sibling distance=1cm},
level 3/.style={sibling distance=1.5cm, level distance=2.5cm}]

\Yvcentermath1

\tikzstyle{every node}=[draw]

\node [shape=circle] (Root) {$\tiny\yng(2,1)$}

child {
	node [shape=rectangle, minimum height=0.9cm] {$\tiny\yng(2) ~\oplus~ \yng(1)~$}
	child {
		node [shape=ellipse] {$\tiny\yng(1,1) ~\oplus~ \yng(1)$}
		child {
			node [draw=none] {$\tiny \yng(1,1,1,1,1)$}
		}
		child {
			node [draw=none] {$\tiny \yng(2,1,1,1)$}
		}
		child {
			node [draw=none] {$\tiny \yng(2,2,1)$}
		}
		child {
			node [draw=none] {$\tiny \yng(3,2)$}
		}
	}
};

\end{tikzpicture}
\end{center}

\end{example}

\noindent We can now formally describe the decomposition puzzle and their solutions.

\begin{definition}
A solution $\mathbf{s}$ to a $(\mu, \lambda)$ \textbf{decomposition puzzle} is a pairing,
\[
(\mu_1, \ldots, \mu_k) \smile (l_{i_1}, \ldots, l_{i_k})
\]
of a good $\mu$-decomposition such that $\lambda$ appears in $(\mu_1, \ldots, \mu_k) \oast (l_{i_1}, \ldots, l_{i_k})$.\\

\begin{definition}
We say a solution \textbf{contributes},
\[
\Contrib(\mathbf{s}):= \alpha\cdot\beta,
\]
where $\alpha$ is the iterated Littlewood-Richardson coefficient $L^\mu_{\mu_1\cdots\mu_k}$ and $\beta$ is the multiplicity with which $\lambda$ appears in the assembly $(\mu_1, \ldots, \mu_k) \oast (l_{i_1}, \ldots, l_{i_k})$.
\end{definition}

\noindent Let $\Sigma = \Sigma_{(\mu, \lambda)}$ denote the set of all distinct solutions to $(\mu, \lambda)$ decomposition puzzles.
\end{definition}

\begin{theorem} \label{TheoremPuzzles}
The coefficient $c_{\lambda\mu}$ is the weighted sum of all solutions to $(\mu, \lambda)$ decomposition puzzles. That is,
\[
c_{\lambda\mu} = \sum_{\mathbf{s} \in \Sigma_{(\mu, \lambda)} } \Contrib(\mathbf{s}).
\]
\end{theorem}

\begin{proof}
From (\ref{EqCoeffsFromFreeLie}), we have that,
\[
\SF{\mu}(\mathcal{L}(V)) \cong \bigoplus_{\lambda} c_{\lambda\mu} \SF{\lambda}(V).
\]
A basic property of Schur functors $\SF{\mu}$ is that,

\begin{equation} \label{EqSchurSum}
\SF{\mu}(A \oplus B) \cong \bigoplus_{\mu_1, \mu_2} L^\mu_{\mu_1\mu_2} \SF{\mu_1}(A) \otimes \SF{\mu_2}(B),
\end{equation}

\noindent where $|\mu| = |\mu_1| + |\mu_2|$ (see, for example, \cite{FH}). It follows from our decomposition of the free Lie algebra into its Lie pieces in (\ref{EqLiePiecesDecomp}), and by iterative applications of (\ref{EqSchurSum}), that,

\begin{equation} \label{EqIterLRFreeLie}
\SF{\mu}(\mathcal{L}(V)) \cong \bigoplus L^{\mu}_{\mu_1\cdots\mu_k} \cdot  \SF{\mu_1}(\SF{l_{i_1}}(V) ) \otimes \cdots \otimes \SF{\mu_k}(\SF{l_{i_k}}(V)),
\end{equation}

\noindent where the sum is over all pairings of all $\mu$-decompositions.\\

\noindent By definition the coefficient $c_{\lambda\mu}$ is the multiplicity with which $\SF{\lambda}(V)$ appears in $\SF{\mu}(\mathcal{L}(V))$. Consider a summand appearing in the RHS of (\ref{EqIterLRFreeLie}) indexed by a pairing,
\[
(\mu_1, \ldots, \mu_k) \smile (l_{i_1}, \ldots, l_{i_k}).
\]
This pairing is a solution to a $(\mu, \lambda)$ decomposition puzzle if and only if $\lambda$ appears as a summand in the assembly of the pairing. Moreover, it is easy to see that the multiplicity with which $\lambda$ appears in this summand is precisely the contribution of that solution.
\end{proof}

\noindent We can immediately say something about coefficients $c_{\lambda\mu}$ when $|\lambda| = |\mu|$.

%\noindent We observe from definition of a decomposition puzzle, that the coefficients $c_{\lambda\mu}$ when $|\lambda| = |\mu|$ are easy to compute.

\begin{lemma} \label{LemmaMuSizeLambda}
Let $\lambda, \mu$ partitions such that $|\lambda| = |\mu|$. Then,
\[
c_{\lambda\mu} = \left\{\begin{array}{ccc}1 &  & \lambda = \mu \\0 &  & \mbox{else}\end{array}\right.
\]
\end{lemma}

\begin{proof}
Let $(\mu_1, \ldots, \mu_k)$ be a $\mu$-decomposition. By Lemma \ref{LemmaAssemblyNoDec}, the size of partitions in an assembly is,
\[
|\mu_1|\cdot|l_{i_1}| + \cdots + |\mu_k|\cdot|l_{i_k}|.
\]
Furthermore, we have that $|\mu_1| + \cdots + |\mu_k| = |\mu|$. Observe that there is only one Lie piece of size 1, namely,
\[
\Yvcentermath1
l_1 = \tiny\yng(1) ~,
\]
and so the only way to obtain partitions of size $|\mu|$ in the assembly is if $k=1$ and $l_{i_1}=l_1$. There is only one $\mu$-decomposition of length $1$, $\mu$ itself! The result follows.
\end{proof}

\noindent With this result in hand we have a potential strategy for computing the coefficients $c_{\lambda\mu}$, namely, enumerate all possible solutions to $(\mu, \lambda)$ decomposition puzzles. The problem, as we outline below, is that the naive approach is computationally infeasible. In the next section we highlight the source of this infeasibility, and provide a workaround that considers the \emph{shape} of a decomposition.

\subsection{Shape analysis} \label{SubSecShapeAnal}

Before we define the shape of a decomposition, we outline the the problem it seeks to address. Fix partitions $\mu, \lambda$. By Theorem \ref{TheoremPuzzles}, our strategy for computing $c_{\lambda\mu}$ is to find all solutions to $(\mu, \lambda)$ decomposition puzzles. Fix a $\mu$-decomposition $(\mu_1, \ldots, \mu_k)$. \emph{A priori}, finding corresponding solutions involves checking the assemblies of all pairings $(l_{i_1}, \ldots, l_{i_k})$ in $\mathbb{L}$. As stated this problem is not even finite! Of course, we don't need to consider all of $\mathbb{L}$. By Lemma \ref{LemmaAssemblyNoDec} we need only consider Lie parts of size at most $|\lambda| = d$, so we can restrict our search to the truncation $\mathbb{L}_{\leq d}$.\\

\noindent Our problem is now finite, but it is too large! Indeed, we are left to check all possible ordered $k$-tuples in $\mathbb{L}_{\leq d}$. For each such pairing we form an assembly, which involves computing $k$ plethysms and $(k-1)$ tensor products. All together, the number of computations for the single $\mu$-decomposition $(\mu_1, \ldots, \mu_k)$ is

\begin{equation} \label{EqTooManyComputes}
\mathcal{O}\left( \frac{f(d)!}{(f(d)-k)!} \cdot k^2 \right),
\end{equation}
where $f: \mathbb{N} \rightarrow \mathbb{N}$ is the function taking $d \mapsto |\mathbb{L}_{\leq d}|$.

\begin{remark} \label{RemarkTwoProblems}
There are two major problems with (\ref{EqTooManyComputes}).

\begin{enumerate}
\item \label{ItemPleth}
(\ref{EqTooManyComputes}) represents the number of plethysm and tensor products that need to be computed - and these operations (especially the plethysm) are computationally expensive.

\item \label{ItemExplode}
The function $f$ grows very quickly (see Remark \ref{RemarkLieGrowsFast}), causing the factorial expression to explode.
\end{enumerate}
\end{remark}

\noindent We address each of these points in turn in the next two sections.

\subsubsection{Avoid unnecessary plethysms and tensor products} \label{SubSubSecAvoidPleth}

The following proposition follows immediately from Lemma \ref{LemmaAssemblyNoDec} and provides a workaround to Remark \ref{RemarkTwoProblems} (\ref{ItemPleth}).

\begin{proposition}
Fix partitions $\mu, \lambda$. If $\mathbf{s} = (\mu_1, \ldots, \mu_k) \smile (l_{i_1}, \ldots, l_{i_k})$ is a solution to the $(\mu, \lambda)$ decomposition puzzle, then,

\begin{equation}\label{EqTargetSizeIsLambda}
|\mu_1| \cdot |l_{i_1}| + \cdots + |\mu_k| \cdot |l_{i_k}| = |\lambda|.
\end{equation}

\end{proposition}

\noindent Notice that this condition can be checked \emph{without computing plethysms or tensor products}. Our modified strategy therefore is only to check assemblies of pairings for which (\ref{EqTargetSizeIsLambda}) holds.

\begin{definition}
The \textbf{shape} of $\mu$-decomposition $(\mu_1, \ldots, \mu_k)$ is the partition $\theta \vdash |\mu|$ with parts given by the sizes of its constituent partitions $\mu_j$. That is,
\[
\theta = (|\mu_1|, |\mu_2|, \ldots, |\mu_k|)
\]
(possibly after reordering). See Fig. \ref{FigPairings}.
\end{definition}

\noindent The figure below depicts the simplification this analysis affords us.

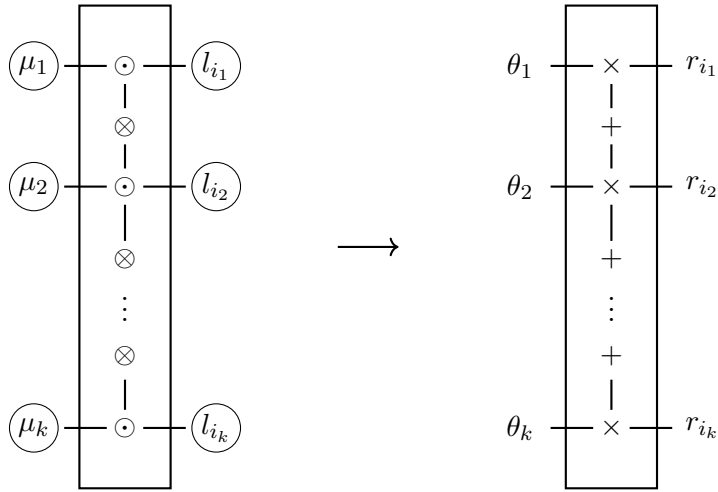
\begin{figure}
\begin{center}
\begin{tikzpicture}[scale=0.8]
\Yvcentermath1

% mus
\draw (0,6) circle (0.4) node {$\mu_1$};
\draw (0,4) circle (0.4) node {$\mu_2$};
\draw (0,0) circle (0.4) node {$\mu_k$};

% lies
\draw (3,6) circle (0.4) node {$l_{i_1}$};
\draw (3,4) circle (0.4) node {$l_{i_2}$};
\draw (3,0) circle (0.4) node {$l_{i_k}$};

% arrows - horizontal
\draw[thick] (0.5,6) -- (1.2,6);
\draw[thick] (1.8,6) -- (2.5,6);
\node at (1.5, 6) {$\odot$};

\draw[thick] (0.5,4) -- (1.2,4);
\draw[thick] (1.8,4) -- (2.5,4);
\node at (1.5, 4) {$\odot$};

\draw[thick] (0.5,0) -- (1.2,0);
\draw[thick] (1.8,0) -- (2.5,0);
\node at (1.5, 0) {$\odot$};

% arrows - vertical
\draw[thick] (1.5, 5.7) -- (1.5, 5.3);
\draw[thick] (1.5, 4.7) -- (1.5, 4.3);
\node at (1.5, 5) {$\otimes$};

\draw[thick] (1.5, 3.7) -- (1.5, 3.1);
\node at (1.5, 2.8) {$\otimes$};

\node at (1.5,2.1)  {$\vdots$};

\draw[thick] (1.5, 0.8) -- (1.5, 0.3);
\node at (1.5, 1.2) {$\otimes$};

% box
\draw[thick]  (0.75,-1) -- (2.25,-1) -- (2.25,7) -- (0.75,7) -- cycle;

\draw[->, thick] (5, 3) --  (6,3);

\begin{scope}[xshift=8cm]
% mus
\node at (0,6) {$\theta_1$};
\node at (0,4)  {$\theta_2$};
\node at (0,0) {$\theta_k$};

% lies
\node at (3,6)  {$r_{i_1}$};
\node at (3,4)  {$r_{i_2}$};
\node at (3,0)  {$r_{i_k}$};

% arrows - horizontal
\draw[thick] (0.5,6) -- (1.2,6);
\draw[thick] (1.8,6) -- (2.5,6);
\node at (1.5, 6) {$\times$};

\draw[thick] (0.5,4) -- (1.2,4);
\draw[thick] (1.8,4) -- (2.5,4);
\node at (1.5, 4) {$\times$};

\draw[thick] (0.5,0) -- (1.2,0);
\draw[thick] (1.8,0) -- (2.5,0);
\node at (1.5, 0) {$\times$};

% arrows - vertical
\draw[thick] (1.5, 5.7) -- (1.5, 5.3);
\draw[thick] (1.5, 4.7) -- (1.5, 4.3);
\node at (1.5, 5) {$+$};

\draw[thick] (1.5, 3.7) -- (1.5, 3.1);
\node at (1.5, 2.8) {$+$};

\node at (1.5,2.1)  {$\vdots$};

\draw[thick] (1.5, 0.8) -- (1.5, 0.3);
\node at (1.5, 1.2) {$+$};

% box
\draw[thick]  (0.75,-1) -- (2.25,-1) -- (2.25,7) -- (0.75,7) -- cycle;
\end{scope}

\end{tikzpicture}
\caption{On the left we depict a (straightened) pairing of $(\mu_1, \ldots, \mu_k)$ with a collection of Lie pieces $(l_{i_1}, \ldots, l_{i_k})$. On the right is the associated shape partition $\theta$ together with the sizes $r_i$ corresponding to the Lie pieces $l_i$. Notice that plethysms and tensor products on the LHS become multiplications and additions on the RHS (resp.).}
\label{FigPairings}
\end{center}
\end{figure}

\paragraph{Strategy.} Our strategy will be to restrict attention to those pairings satisfying (\ref{EqTargetSizeIsLambda}). We describe the algorithm producing such pairings in Algorithm \ref{AlgBuildInstructions}. Observe that it is possible for two different $\mu$-decompositions to have the same shape. It is therefore more efficient to find solutions to (\ref{EqTargetSizeIsLambda}) among the set of shapes, and to cache these solutions in a hash table,

\begin{equation} \label{EqLookup}
\{ \mbox{shape} : \mbox{indices of Lie pieces} \}.
\end{equation}

\noindent This strategy means we only compute tensor products and plethysms when their target-size is valid. It therefore addresses Remark \ref{RemarkTwoProblems} (\ref{ItemPleth}), as promised.

\begin{example} \label{ExampleSavings}
To illustrate the scale of savings this makes; when $k=3$ and $d=9$, the number of pairings of target size 9 is $148$, whereas the number of possible $3$-element subsets of $\mathbb{L}_{\leq 9}$ is 84027234. Of course as $d$ increases and as $k$ increases this difference only increases!
\end{example}

\subsubsection{Improved upper bound on the size of Lie pieces}

We now address the second problem, Remark \ref{RemarkTwoProblems} (\ref{ItemExplode}). Recall that the source of this problem was that the number of Lie pieces of size $\leq d$ grows very quickly as a function of $d$. Our strategy is to find an improved upper bound on the truncation of Lie pieces.

\begin{definition}
Fix a shape $\theta = (\theta_1, \ldots, \theta_k)$ and a degree $d \in \mathbb{N}$. Define $\varphi = \varphi(\theta, d) \in \mathbb{N}$ by,
\[
\varphi =  \left\lfloor \frac{d - (\theta_1 \cdot r_1 + \cdots + \theta_{k-1} \cdot r_{k-1})}{\theta_k} \right\rfloor
\]
where $r_i = |l_i|$ is the size of the $i$-th Lie piece. 
\end{definition}

\begin{lemma}
Fix partitions $\mu, \lambda$ and a $\mu$-decomposition of shape $\theta$. Let $\varphi = \varphi(\theta, |\lambda|)$. Then any solution to the $(\mu, \lambda)$ decomposition puzzle involving this $\mu$-decomposition can be found in the truncation,
\[
\mathbb{L}_{\leq \varphi}.
\]
\end{lemma}

\begin{proof}
As usual, let $r_i$ denote the size of the $i$-th Lie piece. Consider the set,
\[
X = \{ \rho \in \mathbb{L}_{|\cdot|}: \theta_1\cdot r_1 + \cdots + \theta_{k-1}\cdot r_{k-1} + \theta_k \cdot \rho \leq |\lambda| \}.
\]
First observe that by construction $\varphi = \max(X)$. Since the $\theta_i$'s are weakly decreasing and the $r_i$'s are weakly increasing, it is clear that,
\begin{equation} \label{EqMinimalCombo}
\theta_1 \cdot r_1 + \cdots \theta_k \cdot r_k
\end{equation}
is minimal among $\{\theta_1 \cdot r_{i_1} + \cdots + \theta_k \cdot r_{i_k}: (r_{i_1}, \ldots, r_{i_k}) \in \mathbb{L}_{|\cdot|}^k\}$.\\

\noindent Suppose for a contradiction that there exists a $k$-tuple $(r_{i_1}, \ldots, r_{i_k}) \in \mathbb{L}_{|\cdot|}^k$ with some $r_{i_j} > \varphi$ such that the assembly-size,
\[
\theta_1 \cdot r_{i_1} + \cdots \theta_k \cdot r_{i_k} \leq \lambda.
\]
Let $\sigma \in \SYM_k$ be a(ny) permutation of the sizes $r_{i_j}$ such that $r_{\sigma(i_1)} \leq r_{\sigma(i_2)} \leq \cdots \leq r_{\sigma(i_k)}$. Then our contradictory hypothesis is that $r_{\sigma(i_k)} > \varphi$.\\

\newpage
\noindent We have,

\begin{align}
|\lambda| &\geq \theta_1 \cdot r_{i_1} + \cdots + \theta_{k-1} \cdot r_{i_{k-1}} + \theta_k \cdot r_{i_k} \notag \\
& \geq \theta_1 \cdot r_{\sigma(i_1)} + \cdots + \theta_{k-1} \cdot r_{\sigma(i_{k-1})} + \theta_k \cdot r_{\sigma(i_k)} \notag \\
& \geq \theta_1 \cdot r_{1} + \cdots + \theta_{k-1} \cdot r_{k-1} + \theta_k \cdot r_{\sigma(i_k)} \notag
\end{align}
where the last inequality follows from the minimality of (\ref{EqMinimalCombo}). This shows that $r_{\sigma(i_k)} \in X$, contradicting the maximality of $\varphi$.
\end{proof}

\noindent The upshot of this result is that we can restrict our search for solutions to the smaller set $\mathbb{L}_{\leq \varphi}$. This addresses Remark \ref{RemarkTwoProblems} (\ref{ItemExplode}) as promised.\\

\begin{example}
We demonstrate the scale of improvement afforded by our improved upper bound $\varphi$. Consider the shape $\theta = (2,2,1)$ and the target-size $9$. We see that $\varphi(\theta, 9) = 3$. The number of 3-element subsets of $\mathbb{L}_{\leq 3}$ is 6, whereas the number of 3-element subsets of $\mathbb{L}_{\leq 9}$ is $84027234$.
\end{example}

\subsubsection{Implementation of shape analysis}
We are ready to turn the discussion above into a procedure that we call \emph{shape analysis}.\\

\noindent Fix a target-size $d \in \mathbb{Z}_{>0}$ and a shape $\theta = (\theta_1, \ldots, \theta_k) \vdash m \leq d$. We compute $\varphi = \varphi(\theta, d)$ and then search in $\mathbb{L}_{\leq \varphi}$ for all $k$-tuples $(l_{i_1}, \ldots, l_{i_k})$ such that,
\[
\theta_1\cdot r_{i_1} + \cdots + \theta_k \cdot r_{i_k} = d,
\]
caching the indices $(i_1, \ldots, i_k)$ as we go.

\begin{definition}
We refer to such a $k$-tuple of indices as an \textbf{instruction}.
\end{definition}

\noindent We implement a recursive algorithm computing all instructions for a given shape $\theta$ and target-size $d$. The psuedo-code for this algorithm is given below.\\

%%%%%%%% ALGORITHM %%%%%%%%%%%
\begin{minipage}{.9\linewidth}

\begin{algorithm}[H]

\label{AlgBuildInstructions}

\SetKwFunction{BuildInstructions}{\textsc{build$\_$instructions}}

\SetAlgoLined

\BlankLine

\KwInput{A target-size $d \in \mathbb{Z}_{>0}$ and a shape $\theta \vdash m \leq d$.;

\BlankLine

Due to the recursive nature of the algorithm we also pass an instruction $I$ (default empty array []) and a pointer $p$ (default int 0) as input.}

\BlankLine

\KwResult{ We cache completed instructions along the way in a hash table. }

\BlankLine

Compute upper bound $\varphi = \varphi(\theta, d)$.;

$L \gets \mathbb{L}_{|\leq\varphi|}$;

\BlankLine

\LinesNumbered

\nonl\BuildInstructions{$d,\theta, L, I, p$}

\BlankLine

\nonl (base case)

\uIf{length $\theta[p:] = 1$}{

\For{$l \in L$}{
\uIf{$d = \theta[p] \cdot |l|$}{
Create new instruction $I'$ from instruction by adding index of $l \in \mathbb{L}$.

Cache new instruction $I'$.
}
}
}

\BlankLine

\Else{
$t \gets \theta[p]$

$p \gets p + 1$

\For{$l \in L$}{

$d' \gets d - t \cdot l$

Create new instruction $I'$ from instruction $I$ by adding index of $l \in \mathbb{L}$.

\BuildInstructions{$d',\theta, L\backslash\{l\}, I', p$}

}
}

\caption{Build instructions}

\end{algorithm}

\end{minipage}\\
%%%%%%%% ALGORITHM %%%%%%%%%%%

\noindent This algorithm caches its results in a hash table. We give that a hash a name.

\begin{definition}(Instructions.)
Let $\mathcal{I} = \mathcal{I}(d)$ denote the hash table (of target-size $d$) mapping shapes $\theta$ to the set of instructions computed in Algorithm \ref{AlgBuildInstructions}.
\end{definition}

\noindent Before presenting our algorithm computing composition factors, there is one subtlety that needs to be addressed.

\paragraph{Over counting.} Certain cases arise when we can over count the number of solutions to a $(\mu, \lambda)$-decomposition puzzle. These are best explained by way of an example. Suppose we have a $\mu$-decomposition of shape $[2,2]$ and we have a target-size of 5. In this case we see that there are two instructions $I_1, I_2$:
\[
\Yvcentermath1
I_1 = (1,2) \qquad \rightsquigarrow \qquad \theta_1 \times \left|~ \tiny \yng(1) ~\right| + \theta_2 \times \left|~ \tiny \yng(1,1)~\right| = 5
\]
\[
\Yvcentermath1
I_2 = (2,1) \qquad \rightsquigarrow \qquad \theta_1 \times \left|~ \tiny \yng(1,1)~\right| + \theta_2 \times \left|~ \tiny \yng(1) ~\right| = 5
\]

\noindent In the case that the underlying $\mu$-decomposition is,
\[
\Yvcentermath1
\mu_1 = \tiny\yng(1,1) \qquad \mu_2 = \tiny\yng(2)
\]
then both of these instructions give rise to potential solutions. However, suppose the underlying $\mu$-decomposition is as follows.
\[
\Yvcentermath1
\mu_1 = \tiny\yng(1,1) \qquad \mu_2 = \tiny\yng(1,1)
\]
In this case, both instructions correspond to the same assembly and any solution arises twice as often as it should.\\

\noindent It is easy to see that a solution involving the $\mu$-decomposition $(\mu_1, \ldots, \mu_k)$ is over counted in this way if and only if it contains repeated partitions $\mu_i = \mu_j$. Moreover, we can explicitly calculate the size of the over-count.

\begin{definition}
Let $(\mu_1, \ldots, \mu_k)$ be a $\mu$-decomposition and let $\{ \nu_1, \ldots, \nu_t \}$ be the set of its distinct partitions. Say that $\nu_i$ appears in the $\mu$-decomposition $n_i$ times. Define the \textbf{over-count factor} of $(\mu_1, \ldots, \mu_k)$ as,
\[
\mbox{over}(\mu_1, \ldots, \mu_k) = \left( n_1 ! \cdots n_l ! \right)^{-1}.
\]
\end{definition}

\noindent In the implementation of our algorithm we will account for over counting by computing the over-count factor. Concretely, the contribution of a given solution $\mathbf{s}$ to a $(\mu, \lambda)$-decomposition puzzle with $\mu$-decomposition $(\mu_1, \ldots, \mu_k)$ is multiplied by the over-count factor over$(\mu_1, \ldots, \mu_k)$.

\section{The algorithm} \label{SecAlgorithm}
We are now ready to outline the algorithm computing $c_{\lambda\mu}$. Our strategy is to compute all coefficients $c_{\lambda\mu}$ with $|\lambda|=d$ fixed at once. By Lemma \ref{LemmaMuSizeLambda} we already know the coefficients $c_{\lambda\mu}$ in the case $|\mu| = d$. Our algorithm will therefore iterate through all partitions $\mu$ of size at most $d-1$. Fix a partition $\mu \vdash m < d$.

\paragraph{$\mu$-decompositions.} We first describe how to generate all possible $\mu$-decompositions. Recall from Section \ref{SubSubSecAvoidPleth} that many $\mu$-decompositions can have the same underlying shape $\theta$. Fix a shape $\theta = (\theta_1, \ldots, \theta_k) \vdash m$ and form the product,
\[
M_\theta := \mbox{Partitions$_\mu(\theta_1$)} \times \cdots \times \mbox{Partitions$_\mu(\theta_k$)},
\]
where,
\[
\mbox{Partitions$_\mu(\theta_i$)} := \{ \mu_i \vdash \theta_i : \mu_i \subseteq \mu \}.
\]
Notice that a $k$-tuple in $M_\theta$ is precisely a $\mu$-decomposition of shape $\theta$. We are then left to enumerate the set of distinct $k$-tuples in $M_\theta$, which we denote $X_\theta$. In our implementation we store this in a hash table.

\begin{definition}($\mu$-decomposisions.)
Let $\mathcal{M} = \mathcal{M}(\mu)$ be the hash table (associated to the partition $\mu$) that maps a shape $\theta$ to the set of distinct $\mu$-decompositions $X_\theta$.
\end{definition}

\paragraph{Assembly.} Fix a shape $\theta \vdash m$. Given a $\mu$-decomposition $(\mu_1, \ldots, \mu_k) \in \mathcal{M}[\theta]$ and an instruction $I = (i_1, \ldots, i_k) \in \mathcal{I}[\theta]$ we need to form the assembly,
\[
(\mu_1, \ldots, \mu_k) \oast I := (\mu_1, \ldots, \mu_k) \oast (l_{i_1}, \ldots, l_{i_k}).
\]
This involves applying a sequence of plethysm and tensor product operations\footnote{We implement our algorithm in SAGE, which has optimised implementations of both plethysm and tensor product.}. We then collect all Lie pieces $\lambda$ appearing in $(\mu_1, \ldots, \mu_k) \oast I$, together with their multiplicities $\beta$. Of course, implementing this assembly involves having a representation for the free Lie algebra. 

\begin{definition}(Assembly of instructions.)
Let $A = A(\mu_1, \ldots, \mu_k; I)$ denote the set of tuples $(\lambda, \beta)$ arising in the assembly $(\mu_1, \ldots, \mu_k) \oast I$.
\end{definition}

%%%%%%%% ALGORITHM %%%%%%%%%%%
\begin{minipage}{.9\linewidth}

\begin{algorithm}[H]

\label{AlgCompositionFactors}

\SetKwFunction{BuildInstructions}{\textsc{build$\_$instructions}}
\SetKwFunction{IterLR}{\textsc{iter$\_$lr}}
\SetKwFunction{CompositionFactors}{\textsc{composition$\_$factors}}

\BlankLine

\KwInput{A target-size $d \in \mathbb{Z}_{>0}$.}

\BlankLine

\KwResult{Compute composition factors $c_{\lambda\mu}$ for all partitions $\lambda$ of size $d$.}

\BlankLine

Initialise all coefficients $c_{\lambda\mu} = 0$ for $\lambda \neq \mu$ and $c_{\lambda\lambda} = 1$.

\BlankLine

\For{$m<d$}{

\BlankLine

\For{$\theta \vdash m$}{
$\mathcal{I}[\theta] \gets \BuildInstructions(d, \theta)$
}

\BlankLine

\For{$\mu \vdash m < d$}{

\BlankLine

$\theta \gets$ shape of $\mu$

\uIf{$\theta \in \mathcal{I}$}{

\BlankLine

$\mathcal{M} \gets \mathcal{M}(\mu)$ the hash table of $\mu$-decompositions.

instructions $\gets \mathcal{I}[\theta]$ 

\BlankLine

\For{$(\mu_1, \ldots, \mu_k) \in \text{\emph{decompositions}}$}{

\BlankLine

$\alpha \gets$ \IterLR($\mu, [\mu_1, \ldots, \mu_k]$)

\BlankLine

\uIf{$\alpha > 0$}{
$A \gets A(\mu_1, \ldots, \mu_k; I)$ the assembly with $I$.

\BlankLine

\For{$(\lambda, \beta) \in A$}{
contribution $\gets \alpha \cdot \beta$

over $\gets$ over$(\mu_1, \ldots, \mu_k)$

$c_{\lambda\mu} \mathrel{+}=$ contribution $*$ over
}
}
}
}
}
}

\caption{Compute composition factors of fixed degree.}

\end{algorithm}

\end{minipage}
%%%%%%%% ALGORITHM %%%%%%%%%%%

\section{Data analysis} \label{SecDataAnalysis}

We are now ready to implement our algorithm. The source code for our implementation is publicly available on GitHub\footnote{\url{https://github.com/aminsaied/composition_factors}}. Recall that the coefficient $c_{\lambda\mu}$ can be regarded as the multiplicity of $\SF{\lambda}$ in

\begin{equation} \label{EqThisAgain}
\SF{\mu}(\mathcal{L}(V)).
\end{equation}
As such we are able to use SAGE's symmetric functions libraries to compute the coefficients $c_{\lambda\mu}$ directly from (\ref{EqThisAgain}) (see Section \ref{SecAppendix}). We use this as a baseline against which we can measure the performance of our algorithm (see Section \ref{SecRunningTime}).\\

\noindent The baseline algorithm is only able to compute those composition factors $c_{\lambda\mu}$ where $\lambda$ is of degree at most 5. See Section \ref{SecRunningTime} for running time experiments. The optimisations in our algorithm allow us to extend this considerably and compute all composition factors of degree at most $14$.
%
%\noindent The baseline algorithm was only able to compute composition factors of degree at most $5$. In contrast, our algorithm was able to compute composition factors of degree at most $14$.
\begin{table}[H]
\begin{center}
\begin{tabular}{c|cc}
& Degree & Number of coefficients \\\hline
Baseline & 5 & 324 \\
Our algorithm & 14 & 257,049
\end{tabular}
\caption{Comparison of our algorithm's range of computation against the baseline algorithm using SAGE's built-in methods.}
\label{TabPerform}
\end{center}
\end{table}

\noindent We are therefore able to extend the range of computation by a factor of over 750. In the next section we begin analysis of the coefficients by visualising the data.

\subsection{Visualisations}
In Fig. \ref{FigShowLeq5} we display the data computed by our baseline algorithm. The axes are labelled by partitions and the colour of the square in position $(\mu, \lambda)$ is determined by the coefficient $c_{\lambda\mu}$ (as per the colour-bar on the right of the plot).

%\begin{figure}[H]
%\begin{center}
%\input{figures/cf_display_side_by_side}
%\caption{The plot at the top displays composition factors of degree $\leq 5$ as computed by the baseline algorithm. The plot below displays composition factors of degree $\leq 14$ as computed by our algorithm.}
%\label{FigCFData}
%\end{center}
%\end{figure}

\begin{figure}[H]
\begin{center}
\includegraphics[scale=0.2]{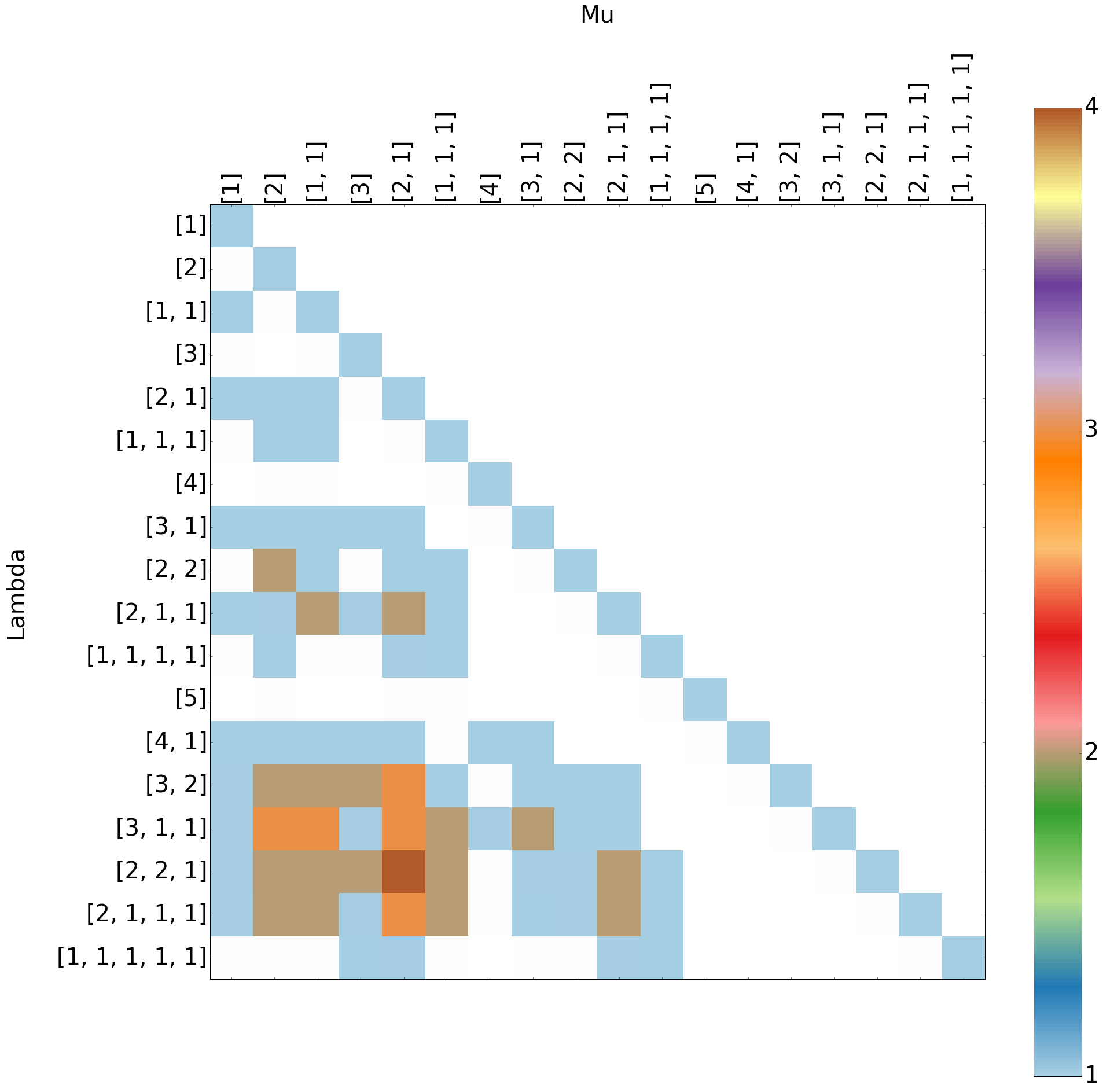}
\end{center}
\caption{Composition factors of degree up to 5. The $(\mu, \lambda)$-entry is coloured according to the coefficient $c_{\lambda\mu}$, with the scale indicated on the right.}
\label{FigShowLeq5}
\end{figure}

\newpage
\noindent We notice some features even from the small amount of data produced by the baseline algorithm.

\begin{enumerate}
\item The diagonal entries are all 1.
\item The matrix is lower-diagonal.
\item If $|\mu| = |\lambda|$ and $\mu \neq \lambda$ then $c_{\lambda\mu} = 0$.
\end{enumerate}

\noindent All of these observations are easy to prove and follow immediately from the definition of $c_{\lambda\mu}$. In short, we don't gain much insight from this plot. In Fig. \ref{FigShowLeq14} we plot for partitions of size $\leq14$. As well as being consistent with the previous observations we now notice some more interesting features.

\begin{figure}[H]
\begin{center}
\includegraphics[scale=0.26]{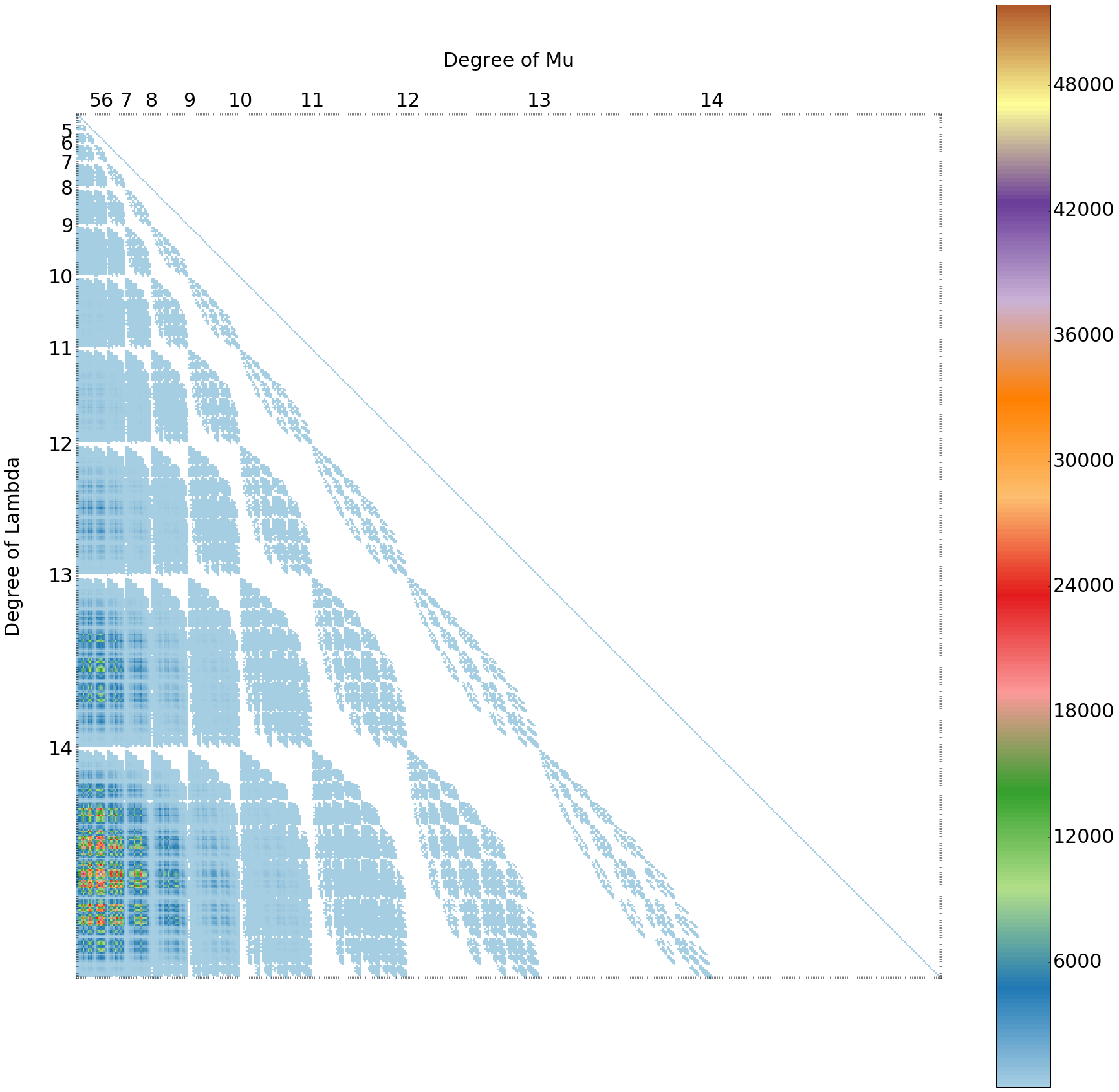}
\end{center}
\caption{Composition factors of degree up to 14. These represents the full range of computations made by our algorithm. For readability we no longer label the partitions on the axes, instead we label the degree (or size) of the partitions at the point at which the degree changes.}
\label{FigShowLeq14}
\end{figure}

%\begin{figure}[H]
%\begin{center}
%\includegraphics[scale=0.18]{images/cf_display_14_clean.png}
%\end{center}
%\caption{Composition factors of degree up to 14. The plot is analogous to that of Fig. \ref{FigShowLeq5}, which itself is represented in the upper left corner in this figure (although, of course, too small to see!). For readability the partitions indexing the coefficients are not shown. Instead we display the point at which the degree of the partitions change.}
%\end{figure}

\newpage
\subsection{Clustering}

At this scale it becomes apparent that there are clusters in the data. The clusters are confined to rectangular blocks determined by sizes of partitions. Concretely, the pair $(\mu, \lambda)$ lies in the same cluster as $(\mu', \lambda')$ if and only if $|\mu| = |\mu'|$ and $|\lambda| = |\lambda'|$. We therefore refer to the cluster containing $(\mu, \lambda)$ as the $(|\mu|, |\lambda|)$-cluster.\\

\begin{figure}[H]
\begin{center}
\vspace{-0.7cm}
\input{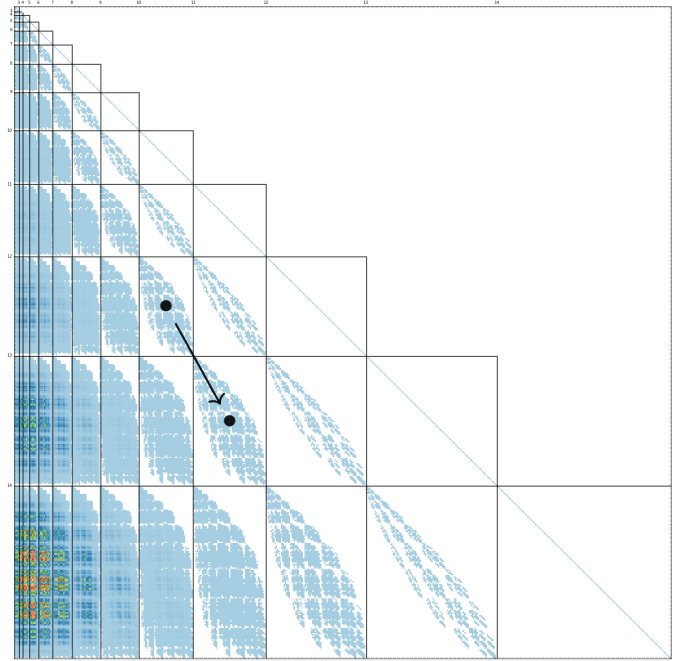}
\vspace{-1.5cm}
\end{center}
\caption{Clusters in the composition factors data. An arrow is drawn between the $(10, 12)$-cluster and the $(11, 13)$-cluster.}
\label{FigureClusters}
\end{figure}

\noindent In particular, the clusters arise as a result of certain rows and columns which are mostly filled with zeros and serve to divide the data into rectangular regions as shown in Fig. \ref{FigureClusters}. Upon inspection we see that these `mostly-zero rows' correspond to coefficients $c_{\lambda\mu}$ where $\lambda$ is of the form $[d]$ or $[1^d]$. This leads us to conjecture that there might be a simple combinatorial rule for determining these coefficients.

\begin{conjecture}~
\begin{enumerate}
\item
Fix $\lambda=[d]$. Then the coefficient $c_{\lambda\mu} = 0$ unless $\mu = [d]$. 

\item 
Fix $\lambda = [1^d]$. Then the coefficient $c_{\lambda\mu} = 0$ unless $\mu$ is of the form $[a, 1^b]$ where $2a + b = d$ or $2a + b = d+1$. 
\end{enumerate}
\end{conjecture}

\noindent A key observation is that the clusters appear to propagate down and to the right. That is, there is a strong similarity between the $(m, d)$-cluster and the $(m+1, d+1)$-cluster. See for example Fig. \ref{FigureClusters}. We investigate this similarity in the next section.

\subsubsection{Stabilising Plateaus}

Fix an initial pair of partitions $\mu, \lambda$ such that $c_{\lambda\mu} > 0$ and consider the process of adding boxes to the top row of these partitions. We introduce some notation.

\begin{definition}
For $n \in \mathbb{N}$, let $\mu^{+n}$ denote the partition obtained from $\mu$ by adding $n$ boxes to the top row of $\mu$.
\end{definition}

\noindent For example,

\[
\Yvcentermath1
\mu ~=~ \yng(2,2,1) ~\rightsquigarrow~ \mu^{+1} ~=~ \yng(3,2,1) ~\rightsquigarrow~ \mu^{+2} ~=~\yng(4,2,1) ~\rightsquigarrow~ \cdots
\]

%\[
%\Yvcentermath1
%\lambda ~=~ \yng(4,3,3,1) ~\rightsquigarrow~ \lambda^{+1} ~=~ \yng(5,3,3,1) ~\rightsquigarrow~ \lambda^{+2} ~=~ \yng(6,3,3,1) ~\rightsquigarrow~ \cdots
%\]

\begin{definition}
Define the \textbf{diagonal push} operation by,
\[
\Delta: (\lambda, \mu) \mapsto (\lambda^{+1}, \mu^{+1}).
\]
\end{definition}

\noindent Notice that if a pair of partitions $(\mu, \lambda)$ is in the $(m,d)$-cluster, then $\Delta(\mu, \lambda)$ lies in the $(m+1, d+1)$-cluster. We are motivated to investigate the behaviour of $c_{\lambda\mu}$ under repeated applications of the operation $\Delta$. Below we plot the sequence of coefficients corresponding to,
\[
(\lambda, \mu), \Delta(\lambda, \mu), \Delta^2(\lambda, \mu), \ldots
\]
for different initial pairs of partitions $\mu, \lambda$.

\begin{figure}[H]
\begin{center}
\includegraphics[scale=0.2]{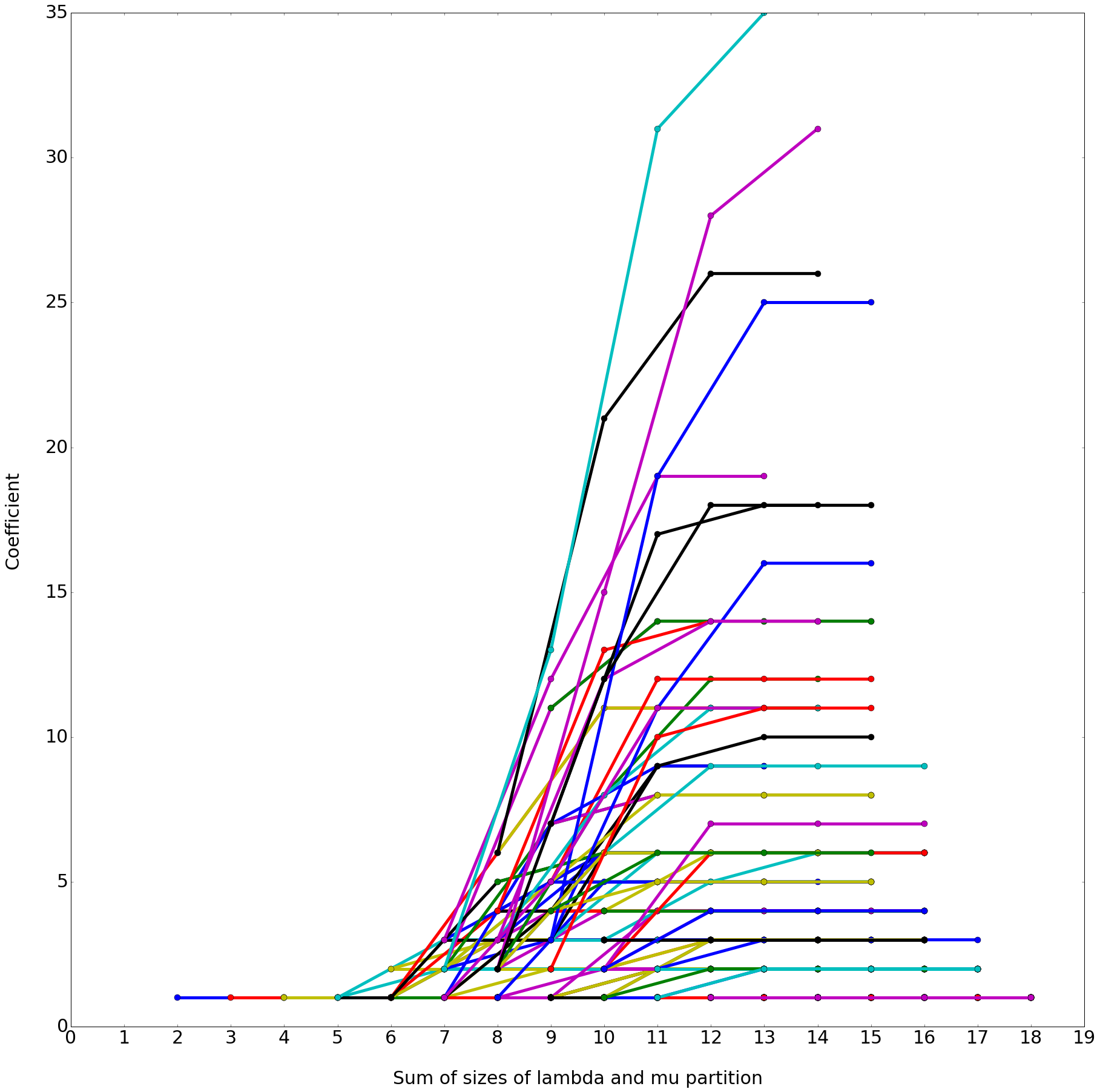}
\end{center}
\caption{The stabilising effect of adding boxes to the top row. Here we plot data from composition factors of size $\leq10$. Notice the distinctive plateaus.}
\end{figure}

\noindent The behaviour is quite striking. Observe that under the operation of $\Delta$, the coefficients rise to a plateau and stabilise. As the sequences progress the data suggests that the coefficients increase to a point, beyond which the sequences flatten into horizontal tails. More formally we make the following conjecture based on these plots.

\begin{conjecture}
Fix partitions $\mu, \lambda$. There exists numbers $x, N$ such that, 
\[
c_{\lambda^{+r} \mu^{+r}} = x,
\]
for all $r \geq N$.
\end{conjecture}

\begin{remark}
In \cite{Amin2} we prove this conjecture by developing \emph{the theory of $\PD$-modules}. We are able to connect these coefficients to the Whitney homology of the lattice of set partitions.
\end{remark}

\subsubsection{Clusters group around the diagonal}
Another pattern that emerges from this perspective is that as clusters move out to the right the data become more concentrated around the major diagonal.

\begin{figure}[H]
\begin{center}
\input{figures/clusters_along_diagonal}
\end{center}
\caption{A sequence of clusters moving to the right.}
\label{FigureClusters}
\end{figure}

\noindent Consider for a moment the upper-right corner of a cluster. One feature the lex-order is that it is a measure of the number of boxes below the first row. We deduce that the upper right region of a cluster corresponds to coefficients $c_{\lambda\mu}$ for which $\mu$ has \emph{more boxes below the first row} than $\lambda$. This observation motivates the following figure, in which we indicate those positions where $\mu$ has more boxes below the first row than $\lambda$ with a black dot.

\begin{figure}[H]
\begin{center}
\hspace{-1cm}\includegraphics[scale=0.16]{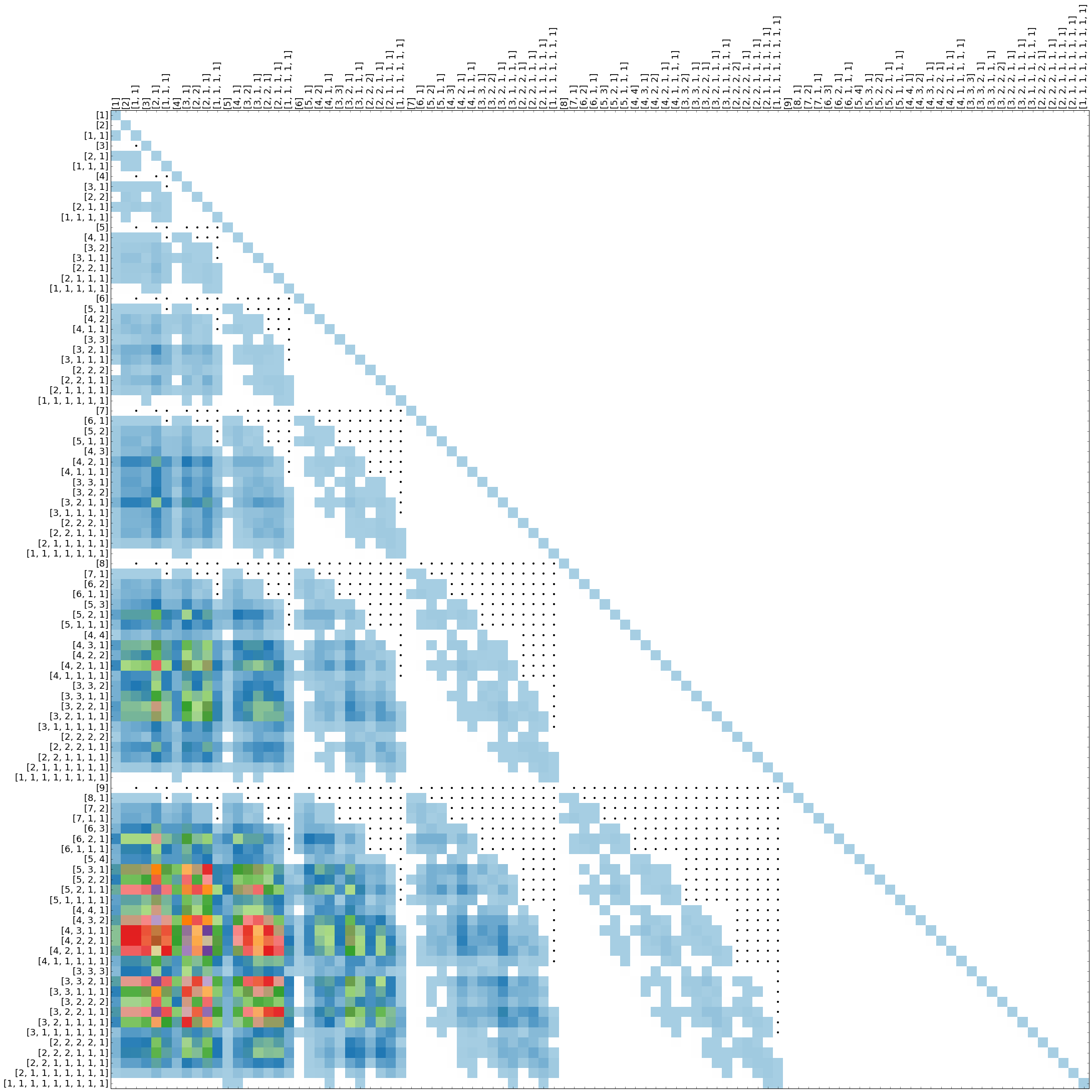}
\end{center}
\caption{Plot displaying when $\mu$ has more boxes below the first row than $\lambda$. Motivation for Conjecture \ref{ConjectureBoxesBelow}}
\end{figure}

\noindent These observations motivate the following conjectures.

\begin{conjecture} \label{ConjectureBoxesBelow}
Let $\lambda, \mu$ be partitions such that $\mu$ has more boxes below the first row then $\lambda$. Then the coefficient $c_{\lambda\mu}=0$.
\end{conjecture}

\begin{remark}
One is tempted to make the symmetrical conjecture regarding ``boxes to the right of the first column". However, upon inspection one sees that is not the case. For example,
\[
\Yvcentermath1
\lambda ~=~ \yng(2,1,1,1) \qquad \mu ~=~ \yng(2,2)
\]
In this case $\lambda$ has one box to the right of the first column, and $\mu$ has two, but,
\[
c_{\lambda\mu} = 1.
\]
A little experimenting with the data leads one to this similar conjecture.
\end{remark}

\begin{conjecture}
Let $\lambda, \mu$ be partitions. Let $m$ be the number of boxes outside the first row and first column in $\mu$, and let $n$ be the number of boxes to the outside the first column in $\lambda$. Then if $m > n$ the coefficient $c_{\lambda\mu}=0$.
\end{conjecture}

\section{Running time experiments} \label{SecRunningTime}

In this section we present the results from an experiment comparing the running times of our algorithm against the baseline algorithm. We first describe our experimental set up. All our running time experiments were performed on computer with a 2.3 GHz Intel Core i7 processor and 16GB RAM. Computations were repeated 10 times and averaged. We compute all composition factors up to degree $d$ using both the baseline and our own algorithm. Here are the corresponding running times (in seconds).

\begin{table}[H]
\begin{center}\begin{tabular}{c|ccccccccc}
Degree & 1 & 2 & 3 & 4 & 5 & 6 & 7 & 8 & 9 \\\hline
Baseline & 0.00189 & 0.009 & 0.107 & 2.35 & 119 & $\infty$ & $\infty$ & $\infty$ & $\infty$ \\
Our algorithm & 0.00237 & 0.0109 & 0.0394 & 0.0979 & 0.239 & 0.703 & 1.82 & 5.64 & 21.9
\end{tabular}
\caption{Running times (in seconds) comparing our algorithm's performance with the baseline.}
\end{center}
\end{table}

\noindent We plot these against a log-scale to account for the differences in running times in seconds.

\begin{center}
\hspace{-35pt} \includegraphics[scale=0.5]{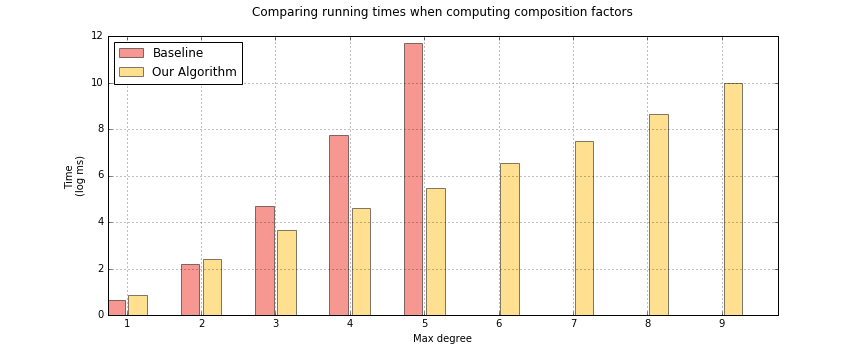}
\end{center}

\noindent Recall that the number of coefficients is increasing rapidly as a function of maximum degree. Below we plot the running times per coefficient. We use the same logarithmic scale in milliseconds.

\begin{center}
\hspace{-35pt} \includegraphics[scale=0.5]{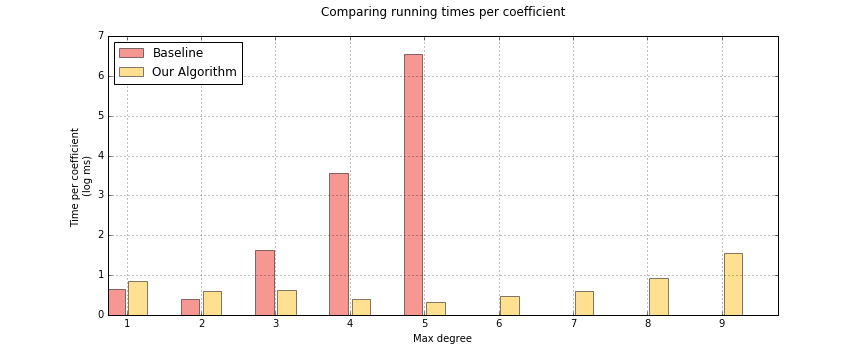}
\end{center}

\subsection{Baseline Algorithm} \label{SecAppendix}
We present the baseline algorithm using SAGE's built-in methods. We first assemble the symmetric function \texttt{lie} corresponding to the truncation $\mathcal{L}_{\leq n}(V)$. For this we use the same implementation for free Lie algebra class \texttt{Lie} (see our source code on GitHub\footnote{\url{https://github.com/aminsaied/composition_factors}}). The key step of this simple algorithm is to compute the plethysm $\SF{\mu}(\mathcal{L}_{\leq d}(V))$, which is implemented in SAGE as follows.

\begin{verbatim}
sage: f = s(mu).plethysm(lie)
\end{verbatim}

\noindent Our baseline algorithm just iterates this over all partitions $\mu \vdash m \leq d$.

%%%%%%%% ALGORITHM %%%%%%%%%%%
\begin{minipage}{.9\linewidth}

\begin{algorithm}[H]

\label{AlgCompositionFactors}

\SetKwFunction{BuildInstructions}{\textsc{build$\_$instructions}}
\SetKwFunction{IterLR}{\textsc{iter$\_$lr}}
\SetKwFunction{CompositionFactors}{\textsc{composition$\_$factors}}

\BlankLine

\KwInput{A target-size $d \in \mathbb{Z}_{>0}$.}

\BlankLine

\KwResult{Compute composition factors $c_{\lambda\mu}$ for all partitions $\lambda$ of size $d$.}

\BlankLine

Initialise all coefficients $c_{\lambda\mu} = 0$.

\BlankLine

$\Lie \gets$ array of Lie pieces of size at most $d$

\BlankLine

\For{$m \leq d$}{

\BlankLine

	\For{$\mu \vdash m$}{
	
\BlankLine
		
		$f \gets \mu \odot \Lie$
		
\BlankLine
		
		\For{$\lambda \in f$}{
		
			$c_{\lambda\mu} \gets c_{\lambda\mu} + 1$
		
		}
	
	}

}

\caption{Baseline algorithm computing $c_{\lambda\mu}$.}

\end{algorithm}

\end{minipage}
%%%%%%%% ALGORITHM %%%%%%%%%%%

\bibliography{/Users/Amin/Dropbox/Math/Latex/coreFiles/aminbib}
\bibliographystyle{acm}

\end{document}